%% file: main.tex
\documentclass[sigconf]{acmart}
\copyrightyear{2026}
\acmYear{2026}
\setcopyright{cc}
\setcctype{by}
\acmConference[FormaliSE '26]{IEEE/ACM 14th International Conference on Formal Methods in Software Engineering}{April 12--13, 2026}{Rio de Janeiro, Brazil}
\acmBooktitle{IEEE/ACM 14th International Conference on Formal Methods in Software Engineering (FormaliSE '26), April 12--13, 2026, Rio de Janeiro, Brazil}
\acmDOI{10.1145/3793656.3793687}
\acmISBN{979-8-4007-2478-7/2026/04}

\input{header}
\title{\tool{}: An SMT-Based Checker for Software Trust Costs}

\begin{document}

\acmBadgeL[https://doi.org/10.6084/m9.figshare.31094197]{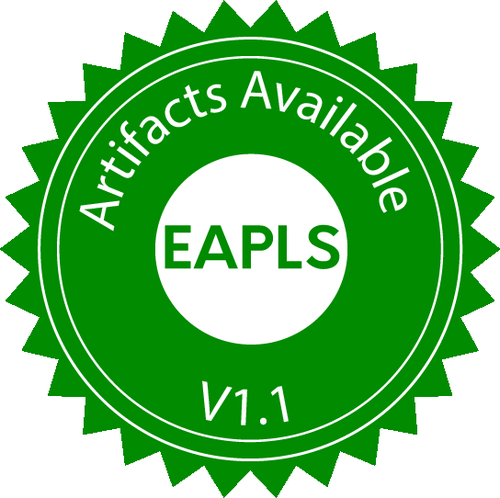}
\acmBadgeL[https://doi.org/10.6084/m9.figshare.31094197]{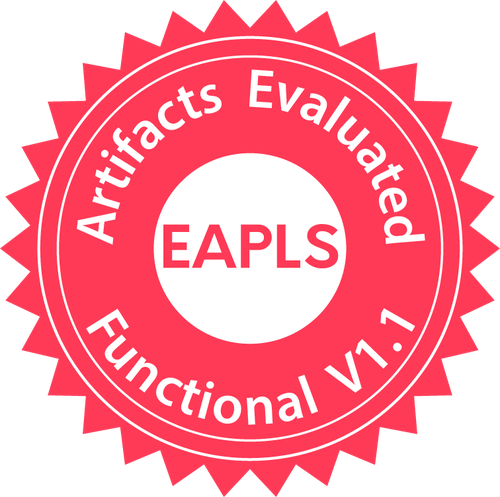}
\acmBadgeL[https://doi.org/10.6084/m9.figshare.31094197]{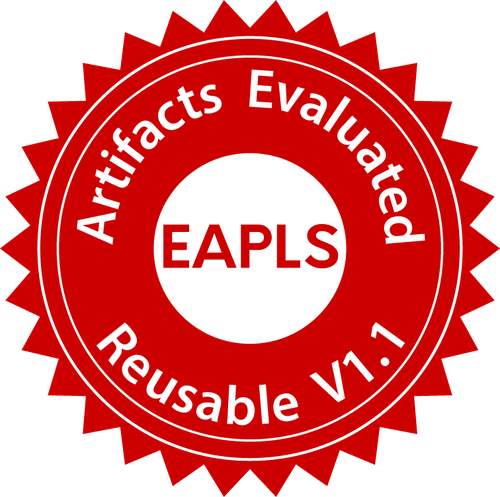}

\author{Muhammad Hassnain}
\orcid{0009-0003-8336-7528}
\affiliation{
    \institution{University of California, Davis}
    \city{Davis}
    \state{CA}
    \country{USA}
}
\email{mhassnain@ucdavis.edu}
\author{Anirudh Basu}
\orcid{0009-0007-5339-9528}
\affiliation{
    \institution{University of California, Davis}
    \city{Davis}
    \state{CA}
    \country{USA}
}
\email{anibasu@ucdavis.edu}
\author{Ethan Ng}
\orcid{0009-0001-8956-4974}
\affiliation{
    \institution{University of California, Davis}
    \city{Davis}
    \state{CA}
    \country{USA}
}
\email{eyng@ucdavis.edu}
\author{Caleb Stanford}
\orcid{0000-0002-8428-7736}
\affiliation{
    \institution{University of California, Davis}
    \city{Davis}
    \state{CA}
    \country{USA}
}
\email{cdstanford@ucdavis.edu}

\renewcommand{\shortauthors}{Hassnain et al.}

\date{\today{}}

\begin{abstract}
Supply chain attacks threaten open-source software ecosystems.
This paper proposes a framework for quantifying trust in third-party software dependencies
that is \emph{formally checkable} -- formalized in satisfiability modulo theories (SMT) --
while at the same time incorporating \emph{human factors}, like the number of downloads, authors, and other metadata that are commonly used to identify trustworthy software in practice.
We use data from software analysis tools and metadata to build a first-order
relational model of software dependencies;
to obtain an overall ``trust cost'' combining these factors, we propose a formalization based on the \emph{minimum trust problem}, which asks for the minimum cost of a set of assumptions that can be used to conclude that the code is safe.
We implement these ideas in \tool{}, targeted for Rust libraries (crates),
incorporating a list of candidate assumptions motivated by quantifiable trust metrics identified in prior work.
Our evaluation shows that \tool{} can identify synthetically generated supply chain attacks and known incidents involving typosquatted and poorly AI-maintained crates.
Its performance scales to Rust crates with many dependencies.
\end{abstract}



\maketitle

\input{1-intro}
\input{2-overview}
\input{3-model}
\input{4-impl}
\input{5-eval}
\balance
\input{6-rw}

\begin{acks}
This research was supported in part by
the US National Science Foundation under award CCF \#2327338.
The authors would like to thank the anonymous reviewers and Prem Devanbu for feedback
leading to improvements to the paper, and Audrey Gobaco for helpful discussions and work on related projects.
\end{acks}

\newpage

\balance

\bibliographystyle{ACM-Reference-Format}
\bibliography{ref}

\end{document}

%% file: header.tex
\usepackage{url}
\hypersetup{
    colorlinks,
    linkcolor={red!40!black},
    citecolor={blue!50!black},
    urlcolor={blue!80!black}
}

\usepackage{amsthm}
\usepackage{cleveref}

\usepackage{subcaption}
\definecolor{Lavender}{rgb}{0.9, 0.9, 0.98}

\theoremstyle{definition}
\newtheorem{theorem}{Theorem}
\newtheorem{definition}{Definition}

\usepackage{algorithm}
\usepackage[noend]{algorithmic}
\newlength{\procedureindent}
\newif\ifprocedureindent
\makeatletter
\@ifundefined{theALC@line}{}{%
  \renewcommand{\theALC@line}{\textcolor{gray}{\scriptsize\arabic{ALC@line}}}%
}
\makeatother
\newcommand{\State}{\STATE}
\newcommand{\Procedure}[2]{\STATE \textbf{procedure} #1(#2)\addtolength{\algorithmicindent}{\procedureindent}\begin{ALC@g}}
\newcommand{\EndProcedure}{\end{ALC@g}\addtolength{\algorithmicindent}{-\procedureindent}}
\newcommand{\ForAll}[1]{\FORALL{#1}}
\newcommand{\EndFor}{\ENDFOR}
\newcommand{\While}[1]{\WHILE{#1}}
\newcommand{\EndWhile}{\ENDWHILE}
\newcommand{\If}[1]{\IF{#1}}
\newcommand{\Else}{\ELSE}
\newcommand{\EndIf}{\ENDIF}
\newcommand{\Output}{\textbf{output} }
\newcommand{\Comment}[1]{\hfill{\itshape\textcolor{gray}{$\triangleright$ #1}}}

\DeclareRobustCommand{\code}[1]{%
  \ifmmode
    \text{\textnormal{\texttt{#1}}}%
  \else
    \textnormal{\texttt{#1}}%
  \fi
}

\newcommand{\tool}{Cargo Sherlock}

\newcommand{\findingsbox}[2]{%
  \bigskip
    \noindent
    \setlength{\fboxsep}{6pt}%
    \fcolorbox{black!60}{blue!10}{%
        \parbox{\dimexpr\linewidth-2\fboxsep-2\fboxrule\relax}{\textbf{#1 Findings:} #2}%
    }%
  \smallskip
}

%% file: 1-intro.tex
\section{Introduction}
\label{sec:intro}

As software ecosystems grow in size and complexity,
automated techniques are needed to prevent supply chain attacks~\cite{ohm2020backstabber,martinez2021software}.
Recent and notable examples of such attacks include the SolarWinds attack~\cite{martinez2021software}; the XZ Utils backdoor~\cite{openwallOsssecurityBackdoor,helpnetsecurityBewareBackdoor}; and for Rust, which is the focus of our implementation, the malicious libraries \texttt{rustdecimal}~\cite{rustsecRUSTSEC20220042Rustdecimal} and \texttt{faster\_log}~\cite{fasterlog-asyncprintln}.
These attacks occur at the \emph{package granularity}, where untrustworthy libraries or software dependencies are added to trustworthy projects, or safe libraries are hijacked by malicious developers.
Many existing techniques can be used to help identify such attacks.
Some of these are \emph{formally checkable} in the sense that they are based on a well-defined formal model of software:
e.g., static analysis tools, fuzzers, type systems, and verification tools.
For such tools, analysis results correspond
to concrete syntactic or semantic source code properties or program traces,
and can therefore be understood in an appropriate formal system~\cite{moller2012static}.
However, these tools are often limited to specific semantic or syntactic properties of programs.

On the other hand, industry developers and security engineers typically rely on a rather different set of criteria in practice:
they use a combination of \emph{trust} in software and \emph{auditing} to ensure that third-party dependencies are safe from software supply chain attacks~\cite{hamer2025trusting,mozilla2024cargovet,cargo-scan-paper,jfrogClosingSupply,crev}.
For example, developers at a large company may be likely to trust software that was developed by that company
or one of its industry partners;
at a smaller scale, an independent developer may be more inclined to trust software that is widely accepted by the community (e.g., has a large number of downloads and stars on GitHub).
Trust and auditing are the foundation of techniques like software bill of materials (SBOM)~\cite{muiri2019framing,hendrick2022state,xia2023empirical,zahan2023software},
which require human-curated or manually generated lists of software components and their purposes.
However, trust and auditing (as currently practiced) often come with no formal guarantees;
there is no formal model over which we can understand the meaning of audits and trust relationships,
and how they compose.

\paragraph{\textbf{Research goals.}}
Is it possible to quantify trust in a way that is formally checkable (in some formal system) yet
also accounts for human factors that many developers rely on in practice?
Our specific focus is to quantify trust in propositional and first-order logic;
roughly speaking, the more assumptions that must be taken for a software dependency ``on faith'', the greater the \emph{trust cost} that we should need to pay in order to believe it functions correctly.
It should be the case that higher trust costs correspond to more risky code.
We identify the following specific goals for such a trust cost:
\begin{itemize}
\item
Trust should be \emph{specific}: a trust cost should commit the analysis to a specific
assumption or set of assumptions about the software
(e.g., the number of downloads, trusted audits, or the results of specific analysis tools).
\item
Trust should be \emph{auditable}: it should be possible for a human to read and understand
the assumptions made by a trust analysis.
\item
Trust should be \emph{configurable}: the user should be able to add
and revise assumptions, e.g., by adding additional tool outputs or human factors as parameters.
\item
Finally, trust should be \emph{composable}: we should be able to compare and combine trust
costs corresponding to different portions of the package ecosystem.
\end{itemize}

We note that traditional software vulnerability analysis and other machine learning (ML)
approaches~\cite{krsul1998software,ghaffarian2017software} are a non-solution to the above.
Such ML techniques would meet our criteria for configurability and composability;
but they would fail the first and second criteria, auditability and specific assumptions that can later
be vetted by human developers and human users.

\paragraph{\textbf{\tool{}: an SMT-based approach.}}
To address these goals, we propose \tool{}, a design and implementation for quantifying supply chain trust
based on a formalization in satisfiability modulo theories (SMT).
\tool{} can be used directly by end users to quantify trust when selecting dependencies to add to a project, and its features are motivated by this intended use case; see \Cref{fig:trust_anyhow} and \Cref{sec:overview}.

We begin with an SMT-based formalization of trust that we call the \emph{minimum trust problem} (\Cref{sec:model}). This particular formalization assumes that we have a list of assumptions that the user might make about the code (such as: ``if the code has a large number of downloads, then it is safe'' or ``if the code is audited by a trusted organization, then it is safe''); each assumption has an associated cost, which can be thought of as the cost inherent in making that assumption.
These assumptions are motivated by prior empirical research;
importantly, they are not necessarily true facts (if they were true unconditionally then trust would not be required!), but rather, they are real thought processes that users might potentially employ to conclude that
dependencies are safe, and are configurable by users.
The minimum trust problem then is to determine the \emph{minimum cost} of a set of assumptions, encoded in propositional logic, that logically entail that the code is safe.
We argue that this model can meet the goals we outlined above; for example, it produces trust costs
that are specific to a particular set of assumptions, can be configured to meet specific users' needs, and are
auditable by human developers.
We consider two algorithms for solving the minimum trust problem, based on a reduction to SMT.
We also show that the problem is NP-hard in the worst case.
This formalization is similar in principle to existing work from the logic programming and databases communities,
for example, the semiring framework for database provenance~\cite{green2017semiring},
assigning weights to logical conclusions in Datalog~\cite{bistarelli2008weighted},
and solvers for weighted Max-SAT~\cite{heras2008minimaxsat,li2009maxsat}.

For the design and implementation (\Cref{sec:impl}), \tool{} gathers software supply chain metadata and the results of automated program analysis tools to build a \emph{first-order relational structure} of the software supply chain -- that is, of all the dependencies of a given package. It then instantiates the minimum trust problem above on a set of possible assumptions on the variables involved in this first-order structure. The implementation also supports ``parameterized assumptions'', where the cost depends in a continuous way on a particular piece of metadata, and ``negative assumptions'', which can be used to incorporate information about software components that are \emph{distrusted}, such as known vulnerabilities or bugs found by the program analysis tools.

Our evaluation (\Cref{sec:eval}) asks four questions:
Can \tool{} be used to prevent synthetically generated supply chain attacks?
Would \tool{} identify risky Rust dependencies (crates) that have been involved in known supply chain incidents?
How does \tool{} fare on known vulnerable crates?
Lastly, does the performance of \tool{} scale to Rust crates with many dependencies?
Among our results, we find that synthetically generated typosquatting attacks generally result in
much higher trust costs.
We also find that of the two algorithms we propose, our more optimized algorithm provides timely
results for crates with a large number of dependencies.

One way to put these results into context is to think about a particular \emph{user model} of how we expect users to interact with our tool.
For example, due to the higher trust costs, if using our tool in the prescribed model,
a user would catch any of the typosquatting attacks considered in our evaluation by identifying and omitting the risky dependency prior to its inclusion.

\paragraph{\textbf{Contributions in summary:}}
\begin{itemize}
\item
A formal model of the supply chain trust problem and two algorithms for its solution
using the SMT theory of linear arithmetic (\Cref{sec:model});
\item
A design and architecture for a usable tool based on these ideas, including concrete input data about Rust code and a collection of built-in candidate assumptions motivated by human factors and available program analyses (\Cref{sec:impl});
\item
An evaluation of the efficacy of our tool for identifying synthetic and real-world supply chain attacks,
and its scalability on Rust libraries with many dependencies (\Cref{sec:eval}).
\end{itemize}

We conclude the paper with related work (\Cref{sec:rw}).
\tool{} is open source (MIT license) and available via a virtual machine image on Figshare\footnote{\url{https://doi.org/10.6084/m9.figshare.31094197}} and as a repository on GitHub.\footnote{\url{https://github.com/davispl/cargo-sherlock}}

%% file: 2-overview.tex
\section{Motivating Overview}
\label{sec:overview}

\begin{figure}[t]
  \centering
  \includegraphics[width=\columnwidth]{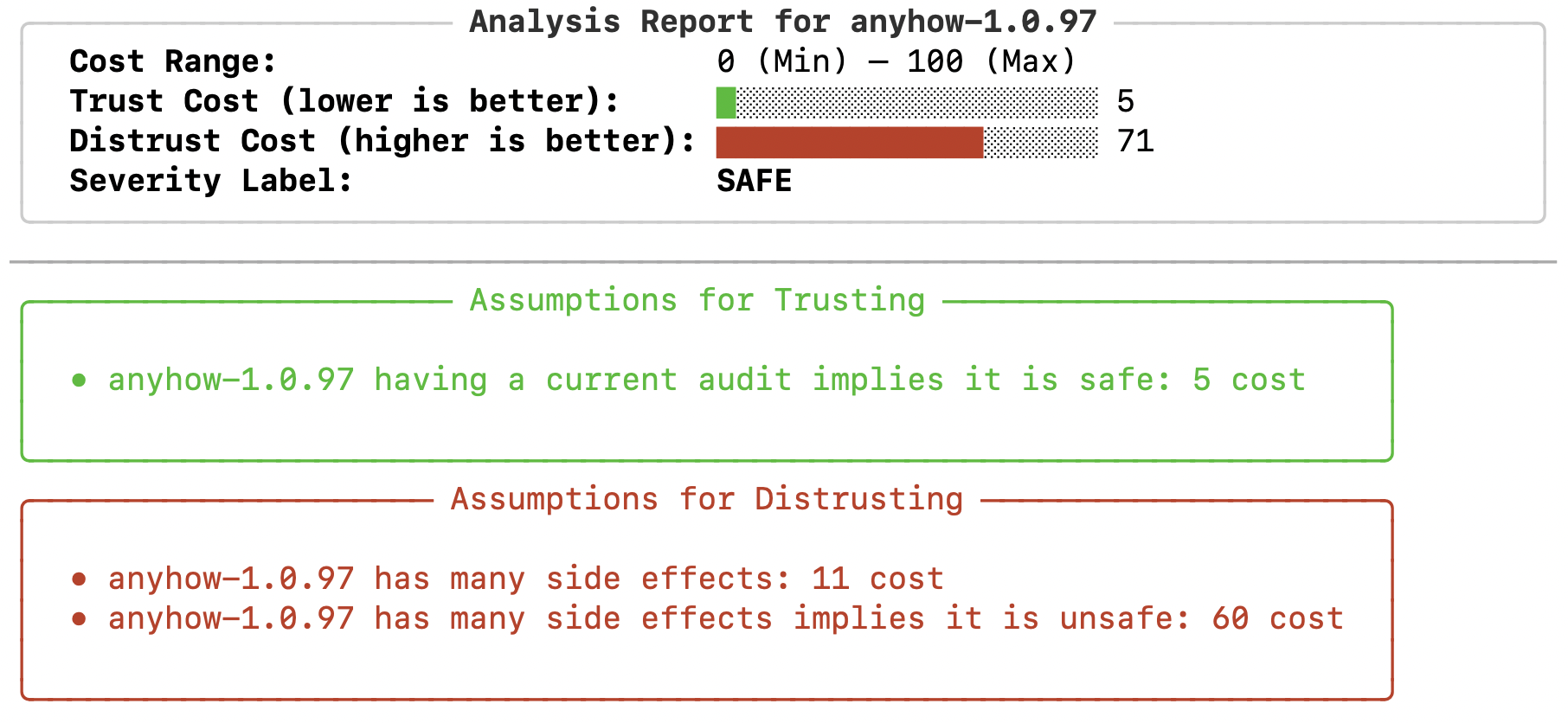}
  \caption{\tool{} example output on the \texttt{anyhow} crate. (Note that lower trust cost, higher distrust cost is better.)
  }
  \Description{Example output of \tool{} on the \texttt{anyhow-1.0.97} crate, showing that it is marked as \texttt{SAFE} due to presence of an audit.}
  \label{fig:trust_anyhow}
\end{figure}

This section presents a motivating overview of \tool{} from a user-facing perspective;
the reader who is more interested in formal details may skip
to \Cref{sec:model}.
We consider a hypothetical example of Ferris, a Rust developer who wishes to install a Rust crate
called \verb|aho_corasick| (implementing the Aho-Corasick algorithm),
and explore how \tool{}'s approach might help Ferris attain the goals laid out in the introduction.

\subsection{Automating trust costs}
After finding the crate \verb|aho_corasick| from a Google search, Ferris faces a choice:
do they trust the crate blindly, or do they first try to determine if the crate is safe to install on their system?
If Ferris is a typical developer, they might employ a common middle ground by relying on various heuristics to
argue that the crate is likely safe:
e.g., that the crate has many downloads, or that it is developed by someone who is trustworthy.
At its core, our tool supports Ferris by attempting to \emph{automate} this reasoning.
The premises of Ferris' arguments are based on \emph{assumptions} about the supply chain ecosystem (e.g., ``the more downloads a crate has, the safer it is'');
\tool{} includes a default list of such assumptions motivated by prior research (see~\Cref{sec:selection_of_assumptions}).
Ferris may find some assumptions more reasonable than others; we mark each assumption with a \emph{cost} to reflect its reasonableness.
If Ferris runs \tool{} on \verb|aho_corasick|, it will present them with \verb|aho_corasick|'s \emph{trust cost}: the total cost of the argument whose premises have the lowest total cost to assume (informally, the \emph{strongest} argument).
After using \tool{}, Ferris determines that \verb|aho_corasick|
is marked as trusted, and that it was marked \texttt{SAFE} due to having a high number of downloads.
See \Cref{fig:trust_anyhow} for an example output;
the trust cost (lower is better) is shown at the top in green.

\subsection{Trust should be specific}
Now suppose that Ferris becomes concerned about trusting crates with a high number of downloads.
To address this concern, Ferris hires an expert auditor, Crabby, who combs through every line of source code of
\verb|aho_corasick|.
Since Ferris trusts Crabby's expertise, Ferris assigns a very low cost to the assumption that
``if the crate is audited by Crabby, it is safe'' -- lower than any other available assumption.
Next, suppose Crabby comes back after the audit and reports that the crate is completely safe to use.
In this scenario, it would seem that the previous assumptions used to mark the crate trusted -- such as its
high number of downloads or the author of the crate -- are no longer relevant to Ferris.
Since \tool{} only reports the \emph{minimum} trust cost for a crate with the associated assumptions,
it does not present other (potentially unhelpful) assumptions.
In other words, Ferris is only presented with a single \emph{specific} reason to trust the crate.
In \Cref{fig:trust_anyhow}, the specific assumptions are shown under ``Assumptions for Trusting'' in green.

\subsection{Trust should be auditable and configurable}
Unfortunately, Ferris later finds out that Crabby is untrustworthy.
Since this renders the existing audit suspect, Ferris revisits the previous assumptions,
and increases the trust cost for ``if the crate is audited by Crabby, it is safe''.
In turn, \tool{} will update its result: the trust cost for \verb|aho_corasick| will increase,
and a new minimum set of assumptions will be identified.
This same mechanism can be used if Ferris wants to add new information to the tool
(e.g., adding additional trusted authors and auditors).

\subsection{Trust should be composable}
During this process, Ferris discovers that \verb|aho_corasick| depends on a second crate, \verb|memchr|,
that was previously overlooked.
Since this could present an additional attack surface, Ferris believes that
the trust cost for \verb|aho_corasick| should be dependent on the cost of its dependencies in some cases.
One way to model this is with assumptions such as the following one:
``if a crate's dependencies are all safe to use, and the crate does not have any code flagged by a static analysis tool, then the crate is safe.''
\tool{} uses a similar assumption to allow arguments that compose trust costs on a crate's dependencies to obtain a trust cost for the crate itself.

\subsection{Incorporating known risks}
Finally, Ferris discovers the RustSec Advisory Database~\cite{rustsec}, which lists vulnerabilities for different Rust crates.
Ferris wants to check whether \verb|aho_corasick| is present in the database.
Unfortunately, nothing we have considered so far accounts for \emph{negative} factors that would decrease trust in a crate; intuitively, this is because we
have only incorporated assumptions that are about trustworthiness.
We thus extend the tool with a \emph{distrust} model to quantify known risks;
we combine the results of the trust and distrust models to obtain a final label for the crate.
In \Cref{fig:trust_anyhow}, the distrust model is shown alongside the trust model with output in red.

%% file: 3-model.tex
\section{Minimum Trust Problem}
\label{sec:model}

In this section, we present the formal model behind \tool{}, which aims to satisfy the requirements from \Cref{sec:intro,sec:overview}.
We define the \emph{minimum trust problem} formally: given a set of candidate assumptions and a conclusion, it asks to find the set of assumptions with minimum cost that entails the conclusion.
We begin with an example, followed by two algorithms for the problem: first, a naive algorithm (\Cref{alg:mintrust}) and second, one based on Horn clause inference (\Cref{alg:horn-mintrust}) that forms the basis for our trust solver implementation.
To conclude the section, we show that the complexity of the minimum trust problem for Horn clauses is NP-complete, justifying such SMT-based approaches.

\subsection{Problem definition}

\paragraph*{Propositional logic basics.}
A \emph{propositional variable} represents a true or false statement that is indivisible into other true or false statements, such as ``the library is safe.'' We use lowercase letters to denote propositional variables.
A \emph{formula} in propositional logic is either a propositional variable, constant ($\top$ for true, $\bot$ for false), or built from propositional variables and logical connectives (``and'' $\wedge$, ``or'' $\vee$, ``not'' $\neg$, ``implies'' $\implies$). We use uppercase letters to denote formulas and uppercase Greek letters to denote sets of formulas.
A \emph{Horn clause} is a formula of the form $(p_1 \wedge \dots \wedge p_n) \implies q$, $\top \implies q$, or $(p_1 \wedge \dots \wedge p_n) \implies \bot$. The antecedent of the conditional is called the \emph{body} and the consequent is called the \emph{head}. A Horn clause is a \emph{definite clause} if it is in the form $(p_1 \wedge \dots \wedge p_n) \implies q$.
A set of definite Horn clauses is \emph{acyclic} if there exists a function $f$ from propositional variables to natural numbers such that for every definite Horn clause $(p_1 \wedge \dots \wedge p_n) \implies q$, $f(q) > f(p_i)$ for all $i$.
A \emph{truth assignment} is an assignment of true or false to propositional variables.
A set of formulas $\Gamma$ \emph{entails} another formula $A$ (written as $\Gamma \models A$) if and only if there is no truth assignment that makes all formulas in $\Gamma$ true and $A$ false. Equivalently, $\Gamma \models A$ if and only if $\Gamma \cup \{\neg A\} \models \bot$.
A set of formulas $\Gamma$ is \emph{satisfiable} if there exists some truth assignment that makes all formulas in $\Gamma$ true.

\begin{definition}[Minimum trust problem]
  Given as input
  (1) a set of propositional formulas $\Gamma=\{A_1, \ldots, A_n\}$ representing assumptions,
  (2) a propositional variable $c$ representing a conclusion, and
  (3) a function $C \colon \Gamma \to \mathbb{N}$ that assigns a cost to making each assumption,
  output a set of assumptions $\Delta \subseteq \Gamma$ such that the sum of the costs of assumptions in $\Delta$ is minimized, subject to the constraint $\Delta \models c$.
  \label{def:minimum-trust}
\end{definition}

We provide a simple example of the minimum trust problem, continuing the example from \Cref{sec:overview}. Suppose that \code{aho\_corasick} depends on \code{memchr}, where \code{aho\_corasick} was developed by Alice, and \code{memchr} was developed by Bob and has no dependencies.
\Cref{tab:example-assumptions} gives a set of candidate assumptions $A_1, \ldots, A_8$ a user can make about the safety of the \code{aho\_corasick} library and their associated costs in this scenario.
We are asked to find a subset of $A_i$ with minimum cost which entails $c$.
In this example, the minimum cost is $45$, incurred from assumptions $a$ (cost 20), $b$ (cost $5$), $(a \land m) \implies c$ (cost 10), and $b \implies m$ (cost $10$).
Note that the set of assumptions is not necessarily unique in general (as there may be a tie for the minimum cost), though it is unique in this example.

\begin{table}[t]
  \centering
  \small
  \setlength{\tabcolsep}{3pt}
  \renewcommand{\arraystretch}{1.15}
  \begin{tabular}{@{}c p{0.55\columnwidth} c c@{}}
    \toprule
    \textbf{No.} & \textbf{Assumption} & \textbf{Encoding} ($A_i$) & \textbf{Cost} \\
    \midrule
    1 & \code{aho\_corasick} is safe. & $c$ & 100 \\
    2 & \code{memchr} is safe. & $m$ & 100 \\
    3 & \code{aho\_corasick} has many downloads. & $d$ & 30 \\
    4 & If \code{aho\_corasick} has many downloads, then it is safe. & $d \implies c$ & 20 \\
    5 & Alice is trustworthy. & $a$ & 20 \\
    6 & Bob is trustworthy. & $b$ & 5 \\
    7 & If Alice is trustworthy and \code{memchr} is safe, then \code{aho\_corasick} is safe. & $(a \land m) \implies c$ & 10 \\
    8 & If Bob is trustworthy, then \code{memchr} is safe. & $b \implies m$ & 10 \\
    \bottomrule
  \end{tabular}
  \caption{Example instance of the minimum trust problem.}
  \Description{Candidate assumptions, their encodings, and costs for the running example.}
  \label{tab:example-assumptions}
  \vspace{-5pt}
\end{table}

\subsection{Naive algorithm}
\Cref{alg:mintrust} shows a naive solution for the minimum trust problem using an SMT solver by encoding the constraints as sentences over linear integer arithmetic and using quantifiers over propositional variables.
To associate each assumption being made to a cost, we use auxiliary propositional variables $a_1, \ldots, a_n$ to track each assumption $A_1, \ldots, A_n$.
Each $a_i$ is assigned to true if and only if $A_i$ is included in the solution set of assumptions $\Delta$.
In the running example, we would have $a_1 \implies c$, $a_4 \implies (d \implies c)$, etc.
We then encode the cost of the assumptions in $\Delta$
as a sum of if-then-else ($ite$) terms;
in the running example, the cost of any particular solution is given by $\code{cost}(a_1, \ldots, a_8) = ite(a_1, 100, 0) + \dots + ite(a_8, 10, 0)$.
The constraint that the conclusion $c$ follows from the assumptions being made, $\Delta \models c$, is encoded as a formula $\code{valid}(a_1, \ldots, a_n)$.
In the running example, $\code{valid}(a_1, \ldots, a_8) = \neg \exists c, m, d, a, b \: ((a_1 \implies c) \wedge \dots \wedge (a_8 \implies (b \implies m)) \wedge \neg c)$.
Finally, we encode using quantifiers the fact that $\code{cost}(a_1, \ldots, a_n)$ is minimized subject to the constraint $\code{valid}(a_1, \ldots, a_n)$.
In the running example, the SMT solver would produce a model that assigns $a_5, a_6, a_7, a_8$ to true and $a_1, a_2, a_3, a_4$ to false, incurring a cost of $45$.

\begin{algorithm}[t]
  \small
  \caption{Naive minimum trust algorithm.}
  \label{alg:mintrust}
  \begin{algorithmic}[1]
    \State \code{library}: a software dependency in a package repository (e.g., a Rust crate on \texttt{crates.io})
    \State \code{get\_all\_assumptions(library)}: a function that returns a set of all possible assumptions for \code{library}
    \State $C$: a function that assigns a cost to an assumption
    \State \code{S}: an SMT solver over the theory of linear integer arithmetic
    \Procedure{JoinAssumptions}{\code{lib}}
      \State $F \gets \top$
      \State $\Gamma \gets \code{get\_all\_assumptions(lib)}$
      \ForAll{$A_i$ in $\Gamma$}
        \State $F \gets F \wedge (a_{i_{\code{lib}}} \implies A_i)$
      \EndFor
      \ForAll{\code{dep} in \code{lib.dependencies}}
        \State $F \gets F \wedge \code{JoinAssumptions}(\code{dep})$ \Comment{\code{JoinAssumptions} is called recursively on each dependency}
      \EndFor
      \State \Output $F$
    \EndProcedure

    \Procedure{SolveMinimumTrust}{}
      \State $F(a_1, \ldots, a_n) \gets \code{JoinAssumptions(library)}$
      \State $\{v_1, \ldots, v_m\} \gets \code{atoms}(F) \setminus \{a_1, \ldots, a_n\}$ \Comment{\code{atoms} returns the set of propositional variables in a formula}
      \State $\code{valid}(a_1, \ldots, a_n) \gets \neg \exists c \exists v_1 \ldots \exists v_m  \: (F(a_1, \ldots, a_n) \wedge \neg c)$ \Comment{$c$ is the conclusion `the library is safe'}
      \State $\code{cost}(a_1, \ldots, a_n) \gets \sum_{i=1}^{n} ite(a_i, C(a_i), 0)$
      \State $\code{minimal} \gets \forall a_1 \ldots \forall a_n \: (\code{valid}(a_1, \ldots, a_n) \implies \code{cost}(a^*_1, \ldots, a^*_n) \leq \code{cost}(a_1, \ldots, a_n))$
      \State $\code{S.solve}(\code{valid}(a^*_1, \ldots, a^*_n) \land \code{minimal})$ \Comment{$(a^*_1, \ldots, a^*_n)$ is the solution to the minimum trust problem}
      \State \Output \code{S.model[$a^*_1, \ldots, a^*_n$], S.model[cost$(a^*_1, \ldots, a^*_n)$]}
    \EndProcedure
  \end{algorithmic}
\end{algorithm}

\begin{theorem}
  \Cref{alg:mintrust} is correct.
  \label{thm:mintrust-correct}
\end{theorem}
\begin{proof}
  \Cref{alg:mintrust} recursively calls \code{JoinAssumptions} on each library in the dependency tree on line 11 to construct the formula $F(a_1, \ldots, a_n)$ that represents the assumptions being made, then uses the above encoding to create the constraint $\code{valid}(a_1, \ldots, a_n)$. The algorithm encodes that the cost of assumptions is minimized subject to the constraint $\code{valid}(a_1, \ldots, a_n)$ on line 18.
\end{proof}

Although it is correct, the use of quantifiers makes the first algorithm computationally expensive;
it generates a linear number of nested quantifiers in the size of the dependency tree,
which have to be expanded before applying a SAT solver.
While some of this complexity is unavoidable (see \Cref{sec:model-complexity}), to address this we next propose an algorithm that does not use quantifiers, and which we have found to be more effective in practice.

\subsection{Horn clause algorithm}
The key idea of \Cref{alg:horn-mintrust} is to restrict the set of assumptions to be acyclic definite Horn clauses;
in this case, the problem can be solved efficiently with a binary search reduction using quantifier-free queries to an SMT solver.
The set of assumptions that can be expressed as acyclic definite Horn clauses is sufficient for all of the assumptions we identify based on prior work (\Cref{tab:trust_model_assumptions}).

In the running example, observe that the set of conjuncts of $F$ that establish the safety of the \code{aho\_corasick} library ($c$) are $a_1 \implies c$, $a_4 \implies (d \implies c)$, and $a_7 \implies ((a \land m) \implies c)$.
Rewritten as Horn clauses with a body and a head, these conjuncts are $a_1 \implies c$, $(a_4 \land d) \implies c$, and $(a_7 \land a \land m) \implies c$.
Therefore, with the goal of showing that $c$ is true, the formula $a_1 \lor (a_4 \land d) \lor (a_7 \land a \land m)$ must evaluate to true.
The only conjunct of $F$ that establishes $d$ is $a_3 \implies d$. Therefore, we can substitute $d$ with $a_3$ in the formula, resulting in $a_1 \lor (a_4 \land a_3) \lor (a_7 \land a \land m)$.
By repeating this unfolding process, we can simplify the formula to $a_1 \lor (a_4 \land a_3) \lor (a_7 \land a_5 \land a_6 \land a_8)$. This formula will evaluate to true if and only if the assumptions made entail the conclusion $c$.
\Cref{alg:horn-mintrust} uses this observation to solve the minimum trust problem when the assumptions are acyclic definite Horn clauses. Due to binary search, \Cref{alg:horn-mintrust} produces a logarithmic number of calls to the SMT solver with respect to the maximum cost, and each call is quantifier-free.

\begin{algorithm}[t]
  \small
  \caption{Horn clause optimized minimum trust algorithm.}
  \label{alg:horn-mintrust}
  \begin{algorithmic}[1]
    \State \code{library}: a library
    \State $C$: a function that assigns a cost to an assumption
    \State \code{S}: an SMT solver over the theory of linear integer arithmetic
    \Procedure{Unfold}{$F$}
      \State $F \gets \code{simplify}(F)$ \Comment{simplify the formula into a conjunction of Horn clauses}
      \State $\Pi \gets \code{clauses}(F)$ \Comment{get the set of Horn clauses in $F$}
      \State $E \gets \bot$ \Comment{$E$ is to be constructed such that $E$ is true iff the assumptions entail the conclusion}
      \ForAll{$\pi_i$ in $\Pi$}
        \If{$\code{head}(\pi_i)$ is the conclusion $c$}
          \State $E \gets E \lor \code{body}(\pi_i)$ \Comment{collect all the bodies of the Horn clauses with head $c$}
        \EndIf
      \EndFor
      \State $V \gets \code{atoms}(F)\setminus \{a_1, \ldots, a_n\}$ \Comment{$V$ stores the set of propositional variables in $F$ other than the $a_i$s}
      \ForAll{$v_i$ in $V$}
        \State $T \gets \bot$ \Comment{the formula to be substituted in place of $v_i$}
        \ForAll{$\pi_i$ in $\Pi$}
          \If{$\code{head}(\pi_i)$ is $v_i$}
            \State $T \gets T \lor \code{body}(\pi_i)$ \Comment{collect all the bodies of the Horn clauses with head $v_i$}
          \EndIf
          \State $E \gets E[v_i / T]$ \Comment{substitute $T$ in place of $v_i$ in $E$}
        \EndFor
        \ForAll{$\pi_i$ in $\Pi$}
          \State $\pi_i \gets \pi_i[v_i / T]$ \Comment{substitute $T$ in place of $v_i$ in $\pi_i$}
        \EndFor
      \EndFor
      \State \Output $E$
    \EndProcedure

    \Procedure{SolveHornMinimumTrust}{}
      \State $F(a_1, \ldots, a_n) \gets \code{JoinAssumptions(library)}$ \Comment{same as in \Cref{alg:mintrust}}
      \State $\code{valid}(a_1, \ldots, a_n) \gets \code{Unfold}(F)$
      \State $\code{cost}(a_1, \ldots, a_n) \gets \sum_{i=1}^{n} ite(a_i, C(a_i), 0)$
      \State $\code{left} \gets 0$ \Comment {initialize the min cost as 0}
      \State $\code{right} \gets \sum_{i=1}^{n} C(a_i)$ \Comment{the max cost is the sum of the costs of all assumptions}
      \While{$\code{left} < \code{right}$}
        \State $k \gets \lfloor (\code{left} + \code{right} + 1) / 2 \rfloor$ \Comment{binary search for min cost}
        \State $\code{cost\_leq\_k} \gets \code{S.solve}(\code{valid}(a_1, \ldots, a_n) \wedge \code{cost}(a_1, \ldots, a_n) \leq k)$
        \If{$\code{cost\_leq\_k}$ is SAT}
          \State $\code{left} \gets k$
        \Else
          \State $\code{right} \gets k - 1$
        \EndIf
      \EndWhile
      \State $\code{min\_cost} \gets \code{left}$
      \State \Output \code{S.model[$a_1, \ldots, a_n$], min\_cost}
    \EndProcedure
  \end{algorithmic}
\end{algorithm}

\begin{theorem}
  \Cref{alg:horn-mintrust} is correct for assumptions consisting of acyclic definite Horn clauses.
  \label{thm:horn-mintrust-correct}
\end{theorem}
\begin{proof}
  \Cref{alg:horn-mintrust} first constructs $E$ as the disjunction of all the bodies of the Horn clauses whose heads are the conclusion $c$ (line 10). The conclusion is provable from the assumptions if and only if $E$ evaluates to true.
  \Cref{alg:horn-mintrust} then proceeds to iterate through $V$ with the goal of eliminating all the
  propositional variables in $V$ from $E$. For each propositional variable $v_i$ in $V$, the
  algorithm constructs a formula $T$ consisting of the disjunction of the bodies of the Horn clauses
  whose heads are $v_i$ and then eliminates $v_i$ by substituting $T$ in place of $v_i$ in $E$,
  similar to the construction of $E$ on line 10. After the substitution, $E$ will still be a
  necessary and sufficient condition for the conclusion to be true.
  The substitution is then applied to each remaining clause in line 19; since the Horn clauses are acyclic, eliminating $v_i$ is guaranteed to remove $v_i$ from all future substitutions. This process is repeated until all propositional variables in $V$ are eliminated from $E$.
  Each step in \code{Unfold} maintains that $E$ is necessary and sufficient for the conclusion to be true.
  The formula generated by \code{Unfold} is used in a binary search to repeatedly query \code{S} for valid minimum trust solutions with different costs until the minimum cost is found.
\end{proof}

It should be noted that the unfolding procedure in \Cref{alg:horn-mintrust} is similar to the well-known SLD resolution algorithm for logic programming~\cite{kowalski1974predicate}.
Unfortunately, the generated formula has an exponential number of clauses in the worst case (with respect to the number of crates in the dependency tree).
This formula, however, is in disjunctive normal form (DNF) after minor simplification, which means the SMT solver can solve each clause independently.
Note that it is also possible to have a reduction that produces a polynomial-size formula, though this is not the reduction we use in our implementation.

\subsection{Complexity}
\label{sec:model-complexity}

To establish the complexity of the minimum trust problem (\Cref{def:minimum-trust}),
we study the corresponding decision problem.
To that end, define MINTRUST to be the problem that asks, given a tuple $(\Gamma, c, C, k)$ --- consisting of a set of assumptions $\Gamma$, a conclusion $c$, a cost function $C$, and integer $k$ --- does there exists a set of assumptions $\Delta \subseteq \Gamma$ such that the sum of the costs of the assumptions in $\Delta$ is less than or equal to $k$ and $\Delta \models c$? Similarly, define HORN-MINTRUST to be the restriction of this problem where $\Gamma$ consists of Horn clauses.
The following theorems establish that the unrestricted minimum trust problem is computationally infeasible (\Cref{thm:mintrust-sigma2p-hard}) and that the Horn-restricted minimum trust problem is intractable but still within the reach of SMT solvers (\Cref{thm:horn-mintrust-np-complete}).

\begin{theorem}
  MINTRUST is $\Sigma_2^P$-hard.
\label{thm:mintrust-sigma2p-hard}
\end{theorem}
\begin{proof}
  We reduce from Shortest Implicant Core (SIC), which is known to be $\Sigma_2^P$-hard~\cite{schaefer2001phcompleteness}~\cite{umans2001implicants}. The SIC problem asks: given a DNF formula $\varphi$, an implicant $I$ of $\varphi$ (a set of literals whose conjunction implies $\varphi$), and an integer $k$, does an implicant $I' \subseteq I$ of size at most $k$ exist?
  Given an instance of SIC with a DNF formula $\varphi$, an implicant $I$ of $\varphi$, and an integer $k$, we construct a MINTRUST instance $(\Gamma, c, C, k)$.
  Let $\Gamma$ contain an assumption for each literal in $I$ with a cost of $1$ and the assumption $\varphi \implies c$ with a cost of $0$.

  By construction, $\Delta \models c$ if and only if $\Delta \models \varphi$ and $(\varphi \implies c) \in \Delta$. Furthermore, $\Delta \models \varphi$ if and only if $\Delta$ contains an implicant $I' \subseteq I$ of $\varphi$. Then, the cost associated with $\Delta$ is equal to the size of $I'$, so $\varphi$ has an implicant of size at most $k$ if and only if the MINTRUST instance has a solution of cost at most $k$. Therefore, an implicant $I' \subseteq I$ of $\varphi$ of size at most $k$ exists if and only if this MINTRUST instance has a solution of cost at most $k$.
\end{proof}
\begin{theorem}
  HORN-MINTRUST is NP-complete.
\label{thm:horn-mintrust-np-complete}
\end{theorem}
\begin{proof}
  \textit{HORN-MINTRUST is NP-hard.}
  We reduce from Vertex Cover~\cite{karp1972reducibility}.
  Given an instance of vertex cover with a graph $G = (V, E)$ and an integer $k$, we construct a HORN-MINTRUST instance $(\Gamma, c, C, k)$.
  Let $\Gamma$ contain an assumption $v_i$ for each vertex $v_i \in V$ with cost $1$, an assumption $v_i \implies e_j$ for each vertex $v_i \in V$ that is incident to edge $e_j \in E$ with cost $0$, and an assumption $(e_1 \land \dots \land e_{|E|}) \implies c$ with cost $0$.
  By construction, $\Delta \models c$ if and only if $\Delta \models e_j$ for all $e_j \in E$ and $((e_1 \land \dots \land e_{|E|}) \implies c) \in \Delta$. Furthermore, $\Delta \models e_j$ for each $e_j \in E$ if and only if $v_i \in \Delta$ and $(v_i \implies e_j) \in \Delta$ for some incident $v_i$. It follows that $\Delta \models e_j$ for all $e_j \in E$ if and only if $\Delta$ is a subset of the $v_i$s for a vertex cover of $G$, together with the assumptions for incident edges. The cost associated with $\Delta$ is equal to the size of the vertex cover, which is at most $k$.

  \textit{HORN-MINTRUST is in NP.}
  We construct a verifier for HORN-MINTRUST. On input a candidate solution $\Delta$, the verifier
  checks that the sum of the costs of the assumptions in $\Delta$ is at most $k$ in one pass.
  Since $\Delta \models c$ holds if and only if the conjunction of the Horn clauses in $\Delta$ with $(c \implies \bot)$ is unsatisfiable, and Horn-satisfiability is in P, checking that $\Delta \models c$ is possible in polynomial time.
\end{proof}

%% file: 4-impl.tex
\section{Design and Implementation}
\label{sec:impl}
\tool{} is a Python-based tool for quantifying trust in the Rust supply chain ecosystem, given as input a Rust crate.\footnote{\url{https://github.com/davispl/cargo-sherlock}} It provides
a logs mode, which collects and displays information about the crate, and a trust mode (see output in \Cref{fig:trust_anyhow}), which is used to compute the trust cost and severity label of the crate,
with components as detailed in the following subsections.

\begin{figure*}[t]
    \centering
    \includegraphics[width=0.72\textwidth]{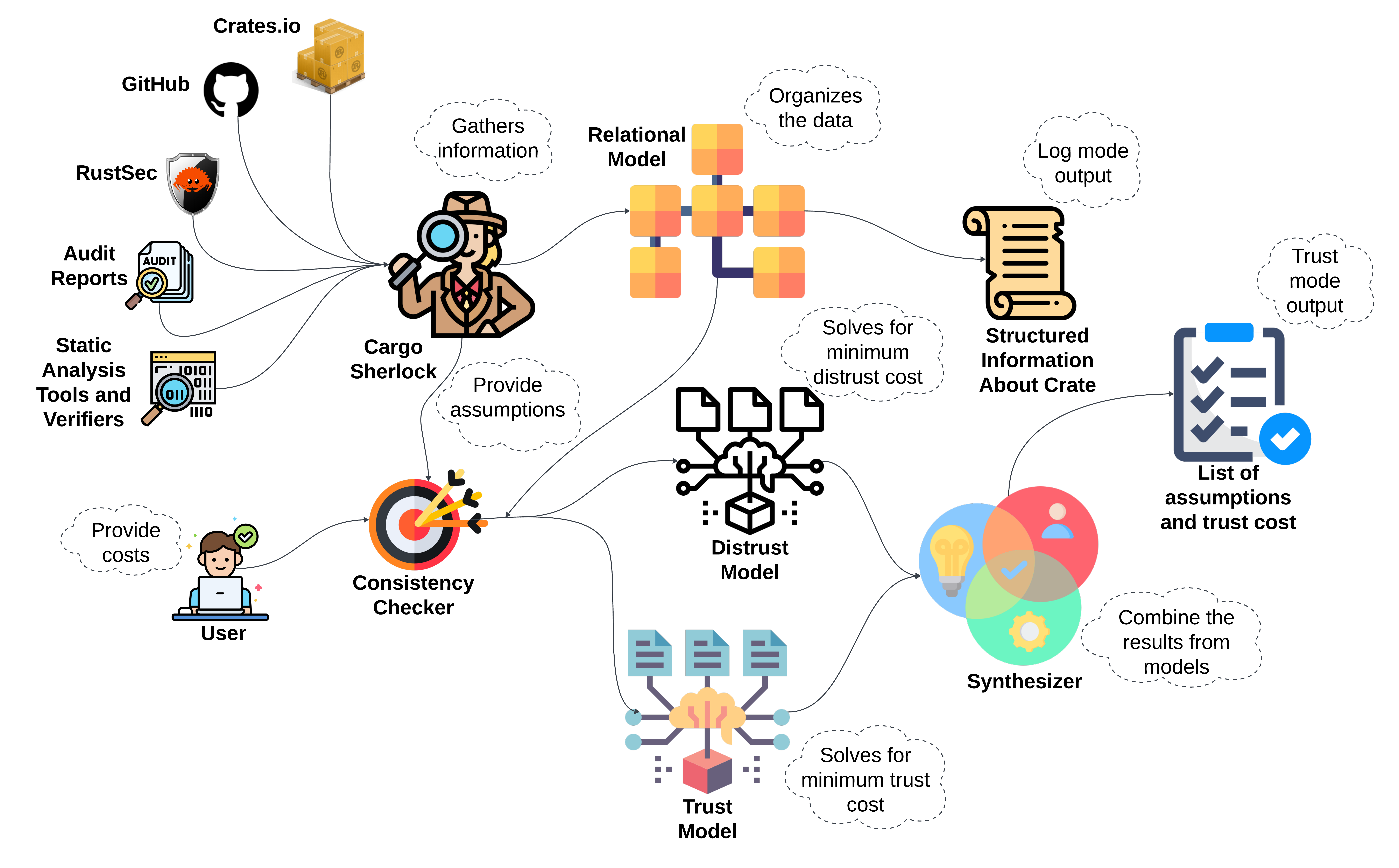}
    \caption{High level overview of \tool{}.}
    \Description{High level overview of the Cargo Sherlock tool. The tool gathers information from multiple sources, builds a relational model, and uses Z3 solver to compute the trust cost and severity label of a given Rust crate.}
    \label{fig:tool}
\end{figure*}

\subsection{Data collection}
\label{sec:data_collection}

\tool{} first gathers crate metadata from multiple sources. From \texttt{crates.io}, it collects information about the author(s), download counts, dependencies, and GitHub repository stars and forks. From Cargo Vet~\cite{mozilla2024cargovet}, it sources audit reports from organizations including Google, Mozilla, and Bytecode Alliance. \tool{} also gathers the outputs from static analysis tools (Miri~\cite{miri} and Cargo Scan~\cite{cargo-scan,cargo-scan-paper}) and vulnerability information from RustSec~\cite{rustsec}.

\subsection{Relational model}
\label{sec:relation_model}
The aggregated data is organized into a relational model as shown in \Cref{fig:er-diagram}. The primary \emph{entities} include: \texttt{Crates}, \texttt{Authors}, \texttt{Audits}, \texttt{Dependencies}, and \texttt{Tool Results}. Each \texttt{Crate} includes basic metadata such as its name, version, downloads, and related git data (e.g., stars and forks). \texttt{Crates} can have multiple associated tool results, audits, authors, and dependencies, following a one-to-many pattern. \texttt{Dependencies} connect crates to other crates, while \texttt{Audits} link organizations and trust criteria to crates.

\begin{figure}[t]
    \centering
    \includegraphics[width=\columnwidth]{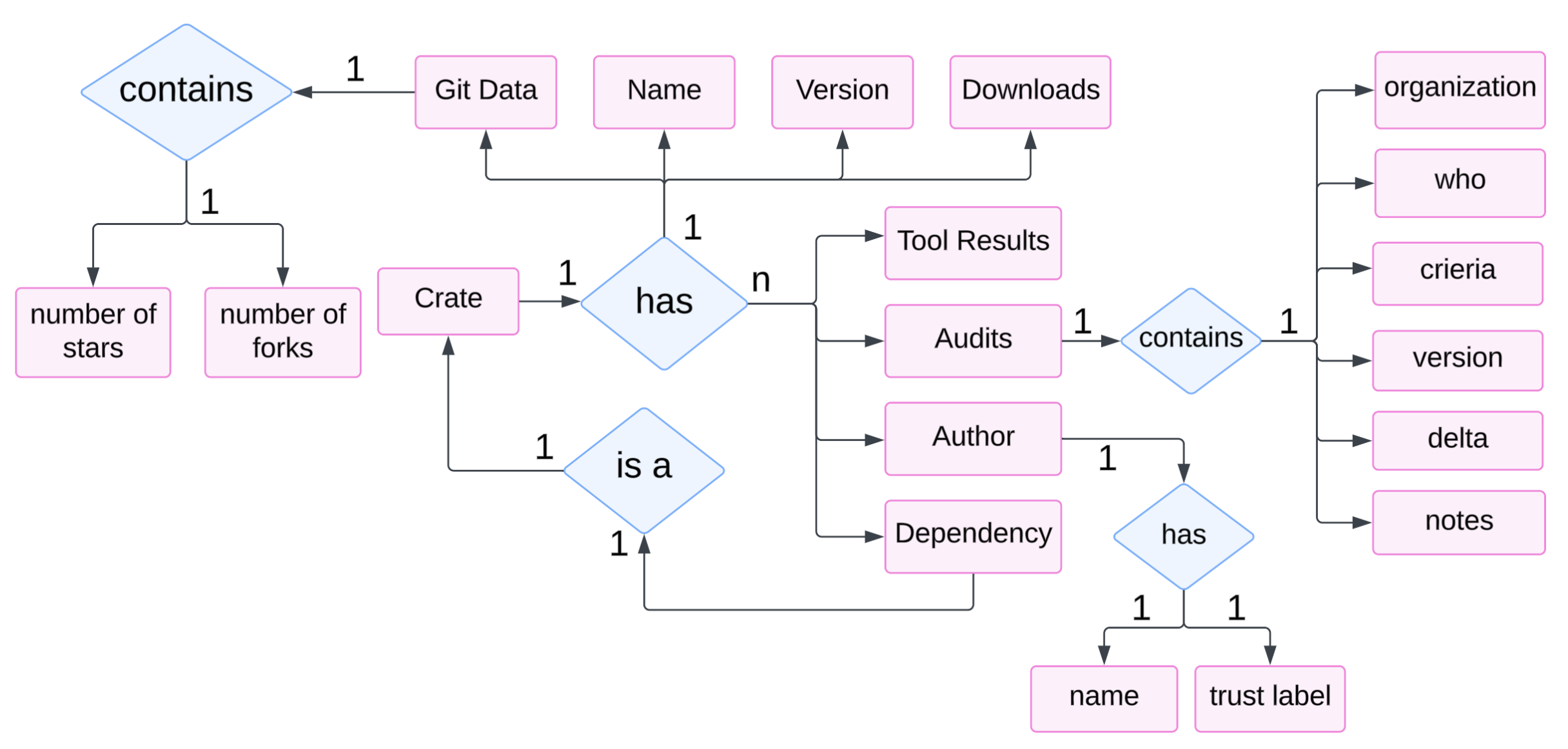}
    \Description{Entity-Relationship (ER) diagram of the Cargo Sherlock relational model, showing entities such as Crate, Author, Audit, Dependency, and Tool Results, along with their attributes and relationships.}
    \caption{Cargo Sherlock entity-relationship (ER) model.}
    \label{fig:er-diagram}
\end{figure}

\subsection{Selection of assumptions}
\label{sec:selection_of_assumptions}

Our candidate assumptions are motivated by prior work on software supply chain trust. Hamer et al.~\cite{hamer2025trusting} identified key trust factors among Rust developers, while Boughton et al.~\cite{boughton2024decomposing} decomposed these into distinct dimensions. \Cref{tab:trust_distrust_split} shows some of the assumptions we support and their justifications in prior research.

\tool{} supports three kinds of assumptions: positive, negative, and parameterized. The costs assigned to
each candidate assumption are configurable by the user. In our implementation, we provide default costs, which we arrived at through a process of trial and error.

\begin{table*}[!t]
  \centering
  \small
  \setlength{\tabcolsep}{3pt}
  \renewcommand{\arraystretch}{0.98}

  \begin{subtable}[t]{0.49\textwidth}
    \centering
    \caption{Trust model assumptions}
    \label{tab:trust_model_assumptions}
    \begin{tabular}{@{}c p{0.66\linewidth} p{0.2\linewidth}@{}}
      \toprule
      \textbf{Cost} & \textbf{Assumption} & \textbf{Justification} \\
      \midrule
      100     & The crate is safe. & Base \\
      25--100 & If the crate has many downloads, it is safe. & \cite{hamer2025trusting} (Tbl.~7) \\
      5       & If the crate version has a passed audit, it is safe. & \cite{hamer2025trusting} (Tbl.~7) \\
      20      & If a past version has a passed audit, it is safe. & \cite{hamer2025trusting} (Tbl.~7) \\
      20--100 & If the repository has many stars/forks, it is safe. & \cite{hamer2025trusting} (Tbl.~7) \\
      10      & If the crate has no side effects and all dependencies are safe, it is safe. & \cite{safe-haskell,cargo-scan-paper} \\
      5       & If the author of the crate is trusted, it is safe. & \cite{hamer2025trusting} (Tbl.~7) \\
      \bottomrule
    \end{tabular}
  \end{subtable}\hfill
  \begin{subtable}[t]{0.49\textwidth}
    \centering
    \caption{Distrust model assumptions}
    \label{tab:distrust_model_assumptions}
    \begin{tabular}{@{}c p{0.66\linewidth} p{0.2\linewidth}@{}}
      \toprule
      \textbf{Cost} & \textbf{Assumption} & \textbf{Justification} \\
      \midrule
      100     & The crate is not safe. & Base \\
      5       & If the crate appears on RustSec with a critical severity label, it is not safe. & \cite{boughton2024decomposing} \\
      90      & If the crate appears on RustSec and the provided crate version is patched, it is not safe. & \cite{boughton2024decomposing} \\
      30      & If the crate is flagged by Miri, it is not safe. & \cite{boughton2024decomposing} \\
      60--100 & If the crate has many side effects, it is not safe. & \cite{safe-haskell,cargo-scan-paper} \\
      10      & If the crate has an unsafe dependency, it is not safe. & \cite{boughton2024decomposing} \\
      \bottomrule
    \end{tabular}
  \end{subtable}

 \caption{Selected assumptions used by \tool{} (including justifications in prior work).}
 \Description{Side-by-side tables summarizing trust and distrust assumptions with costs and justifications.}
  \label{tab:trust_distrust_split}
  \vspace{-5pt}
\end{table*}

\paragraph*{Positive assumptions:}
These are candidate assumptions a user might use to reason that a crate is safe. By default, \tool{} includes the following:

\begin{itemize}
    \item \textbf{Trusted Author}: A crate is safe if it is authored by a well-established developer in the Rust community~\cite{hamer2025trusting}. \tool{} includes a default list of trusted developers, which users can customize.

    \item \textbf{Audit Passed}: A crate is safe if it has passed an audit~\cite{hamer2025trusting}. This assumption has a variant for crates with a previous version that have passed an audit.

    \item \textbf{Side Effects and Safe Dependencies}: A crate is safe if it has no \emph{side effects}, i.e., no observable change to program state or the external environment beyond returning a value, and all of its dependencies are safe~\cite{safe-haskell,cargo-scan-paper}. Side effects are identified using Cargo Scan~\cite{cargo-scan}.

    \item \textbf{Crate is Safe}: By default, a crate is safe (this assumption has the highest possible cost). \tool{} includes this assumption with no justification as a base case.

\end{itemize}

\paragraph*{Negative assumptions:}
These are candidate assumptions a user might use to reason that a crate is unsafe. By default, \tool{} includes the following:

\begin{itemize}
    \item \textbf{RustSec}: A crate is unsafe if it appears in the RustSec advisory database~\cite{rustsec}. This assumption has variants for crates with a previous version that appear in RustSec and for different RustSec labels.

    \item \textbf{Tool Results}: A crate is unsafe if it is flagged by a static analysis tool~\cite{boughton2024decomposing}. Currently, \tool{} supports the assumption that crates flagged by Miri~\cite{miri} are unsafe.

    \item \textbf{Unsafe Dependency}: A crate is unsafe if it has an unsafe dependency.~\cite{boughton2024decomposing}

    \item \textbf{Crate is Unsafe}: By default, a crate is unsafe (this assumption has the highest possible cost). \tool{} includes this assumption with no justification as a base case.
\end{itemize}

\paragraph*{Parameterized assumptions:}
These are candidate assumptions for which the trust cost is calculated as a function of certain metadata parameters; we use a smoothing function so that the trust cost changes gradually with changes to the parameter. We also allow parameterized negative assumptions, which work for distrust cost in the same way. By default, \tool{} includes the following:

\begin{itemize}
    \item
    \textbf{Number of Downloads}: A crate with many downloads is safe~\cite{hamer2025trusting}.

    \item
    \textbf{GitHub Stars and Forks}: A crate with many stars and forks on its GitHub is safe~\cite{hamer2025trusting}.

    \item
    \textbf{Side Effects}: A crate with many side effects is unsafe~\cite{safe-haskell,cargo-scan-paper}.
\end{itemize}

\tool{} applies a \emph{consistency checker} to ensure checks on all assumption
costs; currently, it checks that no assumption has a cost higher than the base cost of 100.

\subsection{Computing the trust cost}
\label{sec:z3_solver}

We use the Z3 SMT solver~\cite{z3} to compute minimum trust and distrust costs using either \Cref{alg:mintrust} or \Cref{alg:horn-mintrust} (default). In our implementation, we replace the explicit binary search over costs with Z3's optimizing SAT interface to obtain the minimum cost directly. The trust cost ranges from 0 to 100,
where a lower trust cost represents a crate that is easier to trust. Similarly, the distrust cost ranges from 0 to 100, where a lower distrust cost represents a crate that is easier to distrust.
The final label is derived by combining both costs, using the function illustrated in \Cref{fig:trust_issues}.

\begin{figure}[t]
  \centering
  \scriptsize
  \includegraphics[width=0.8\linewidth,height=0.6\linewidth,keepaspectratio]{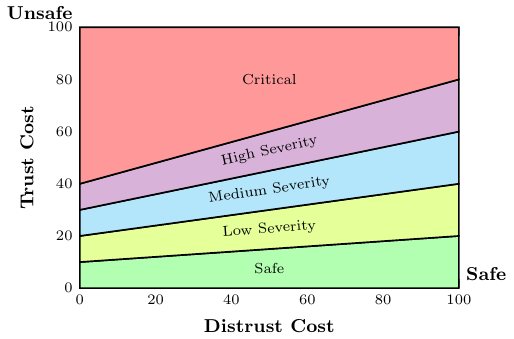}
  \caption{\tool{} combines trust and distrust costs to produce a severity label.}
  \Description{A scatter plot showing the relationship between distrust cost and trust cost, divided into regions labeled Safe, Low Severity, Medium Severity, High Severity, and Critical. The plot illustrates how \tool{} combines these costs to produce a severity label for Rust crates.}
  \label{fig:trust_issues}
\end{figure}

%% file: 5-eval.tex
\section{Evaluation}
\label{sec:eval}

In this section, we evaluate \tool{} on the following research questions:

\begin{itemize}
    \item \textbf{RQ\textsubscript{\textcolor{red}{1}}: Does \tool{} detect synthetically generated typosquatting attacks?}
    \item \textbf{RQ\textsubscript{\textcolor{red}{2}}: Does \tool{} detect real-world supply chain risks?}
    \item \textbf{RQ\textsubscript{\textcolor{red}{3}}: How does \tool{} output fare on known vulnerable crates?}
    \item \textbf{RQ\textsubscript{\textcolor{red}{4}}: How does the performance of \Cref{alg:mintrust,alg:horn-mintrust} scale with the number of dependencies?}
\end{itemize}

\paragraph*{Experimental setup.}
All evaluations were conducted on an M2-Pro MacBook Pro
with 16 GB of RAM.
Our experiments are run on various sets of crates as described below: the 100 most downloaded crates from \texttt{crates.io}, a random sample of 1000 crates, a sample of crates flagged in the RustSec database, and two case studies.
All experiments were run between March and November 2025.

\paragraph*{User model.}
For all experiments, we adopt an idealized ``rational developer'' model in the sense of Dimoulas et al.~\cite{dimoulas2025rational}. The rational developer interacts with our tool as follows: when they intend to add a dependency to their existing Rust project, they execute \tool{} on the selected crate. This generates a report accompanied by a severity label;
we assume that the developer looks at the report and decides whether to go forward with adding the dependency based on the severity label. For each experimental section, we interpret our findings through this user model and discuss the resulting implications for prospective users of \tool{}.

\subsection{RQ\texorpdfstring{$_{\textcolor{red}{1}}$}{1}: Does \tool{} detect synthetically generated typosquatting attacks?}
\label{sec:rq1}

\emph{Typosquatting} is an attack technique that exploits typographical errors to deceive users into downloading a malicious package, such as downloading \texttt{serde\_yml} in place of \texttt{serde\_yaml}. We consider typosquatting because it is a common category of supply chain attack, and aligns closely with \tool{}'s intended use case.
To address RQ1, we constructed synthetic typosquatted versions of the top 100 most frequently downloaded Rust crates from \texttt{crates.io}. We selected popular crates as a test sample because their popularity makes them a lucrative target for an attacker.
We generated typosquatted variants by modifying the name and author of each crate, while keeping the functionality unchanged. To determine how many downloads would be obtained by a lesser-known crate by default, we measured that placeholder crates on \texttt{crates.io} received an average of 171 downloads within 3 months. We use this number for the download count in our synthetic experiment.
After creating typosquatted variants, we evaluated them using \tool{} and compared their
severity labels with corresponding original crates. \Cref{fig:typo_heat} shows the distribution of
obtained severity labels. Most of the original crates are classified as \texttt{safe} (73\%) and no
crate is classified higher than \texttt{low\_severity}. On the other hand, among typosquatted crates,
most are classified as \texttt{critical}. However, 25 typosquatted crates also receive the \texttt{safe} label, likely because they have no
side effects and they all have safe dependencies. The severity label increases for the typosquatted crate when compared to the original crate in all other cases.

\findingsbox{RQ$_{\textcolor{red}{1}}$}{Most typosquatted versions of popular Rust crates incur substantially higher trust costs and higher severity labels compared to their original counterparts when evaluated by \tool{}.
\textbf{Implication:} A user interacting with the tool would identify the typosquatted crate variant when adding a dependency to their code.}

\begin{figure}[!t]
  \centering

  \begin{subfigure}[t]{\columnwidth}
    \centering
    \includegraphics[width=0.6\linewidth]{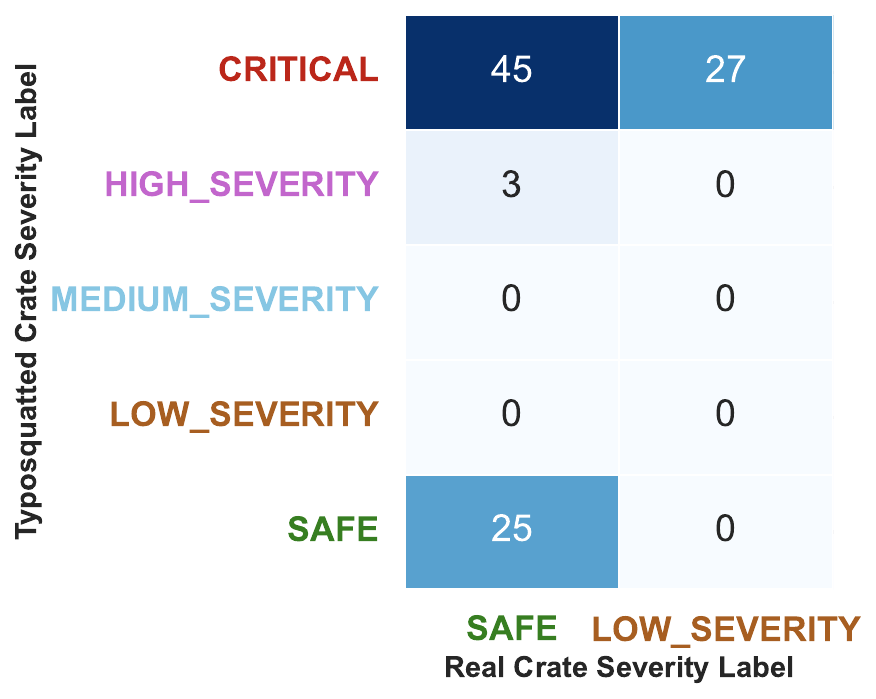}
    \caption{\tool{} results on synthetically generated typosquatting attack crates (compared to their real versions).}
    \Description{\tool{} severity labels on synthetically generated typosquat crates (compared to their real versions).}
    \label{fig:typo_heat}
  \end{subfigure}

  \medskip

  \begin{subfigure}[t]{\columnwidth}
    \centering
    \footnotesize
    \setlength{\tabcolsep}{3pt}
    \renewcommand{\arraystretch}{1.05}
    \begin{tabular}{@{}p{0.36\linewidth} c c p{0.27\linewidth}@{}}
      \toprule
      Crate & Trust Cost & Distrust Cost & Severity Label \\
      \midrule
      \multicolumn{4}{@{}l}{\textbf{Pair 1:} \texttt{fast\_log} (original) vs \texttt{faster\_log} (typosquatted)}\\
      \cmidrule(lr){1-4}
      \texttt{fast\_log-1.7.7} & 37 & 73 & \textcolor{cyan!50!white}{\textbf{MEDIUM\_SEVERITY}} \\
      \texttt{faster\_log-1.7.8}  & 75 & 80 & \textcolor{red!80!black}{\textbf{CRITICAL}} \\
      \midrule
      \multicolumn{4}{@{}l}{\textbf{Pair 2:} \texttt{serde\_yaml} (original) vs \texttt{serde\_yml} (AI-maintained fork)}\\
      \cmidrule(lr){1-4}
      \texttt{serde\_yaml-0.9.33} & \textbf{6} & 70 & \textcolor{green!50!black}{\textbf{SAFE}} \\
      \texttt{serde\_yml-0.0.12}  & 32         & 70 & \textcolor{orange!70!black}{\textbf{LOW\_SEVERITY}} \\
      \bottomrule
    \end{tabular}
    \caption{\tool{} results on real-world incidents.
    }
    \Description{A table summarizing trust cost, distrust cost, and severity level for two pairs of crates: fast_log vs faster_log, and serde_yaml vs serde_yml.}
    \label{tab:trust_summary}
  \end{subfigure}
  \caption{\Cref{sec:rq1,sec:rq2} experimental results.}
  \label{fig:severity_overview}
  \Description{RQ1 and RQ2 severity results overview.}
\end{figure}

\subsection{RQ\texorpdfstring{$_{\textcolor{red}{2}}$}{2}: Does \tool{} detect real-world supply chain risks? }
\label{sec:rq2}

We address this question as a case study on Rust crates that have been involved in known incidents on \texttt{crates.io}. We obtained the source code for two such crates that exist to our knowledge: \texttt{faster\_log} and \texttt{serde\_yml}.\footnote{We are aware of known incidents for at least two other Rust crates:
  \texttt{rustdecimal}~\cite{rustsecRUSTSEC20220042Rustdecimal} and
  \texttt{asyncprintln}~\cite{fasterlog-asyncprintln}.
  Unfortunately, we were unable to obtain the source code for the typosquatted versions of these crates to include them in the analysis.
}

We first examine \texttt{faster\_log}, an identified malicious typosquat of \texttt{fast\_log}~\cite{fasterlog-asyncprintln}. Beyond reproducing the original crate's logging functionality, it injects code that scans log files for sensitive material (including cryptocurrency private keys) and exfiltrates any matches to a remote endpoint. \tool{} assigns a \texttt{MEDIUM\_SEVERITY} label (trust cost $37$, distrust cost $73$) to \texttt{fast\_log}, basing the trust on the author (they have a history of publishing crates) and the distrust on the presence of side effects (it performs file I/O). In contrast, \tool{} assigns \texttt{faster\_log} a \texttt{CRITICAL} label, with a substantially higher trust cost ($75$) and elevated distrust cost ($80$). The change in authorship (the author had no other crates) and side effects leads to a different trust and distrust cost.

Our second case is \texttt{serde\_yml}, a fork of \texttt{serde\_yaml} (the latter authored by David Tolnay, a well-established and trusted developer in the Rust community). After the original crate was marked unmaintained, the fork was published on \texttt{crates.io} and had accumulated over 1{,}057{,}000 downloads by March~19,~2025~\cite{tolnay2024impression}. The fork incorporates a substantial amount of AI-generated code, including low-quality ``AI slop'' changes that were later criticized by Tolnay~\cite{tolnay2024impression}. As shown in Table~\ref{tab:trust_summary}, the original crate receives a \texttt{SAFE} label, while the fork is classified with the higher \texttt{LOW\_SEVERITY}, primarily due to the change in authorship. The distrust costs are similar across the two crates since no identifiable malicious code or side effects were introduced.

\findingsbox{RQ$_{\textcolor{red}{2}}$}{
\tool{} assigns higher severity labels to the malicious typosquat \texttt{faster\_log} and AI-maintained fork \texttt{serde\_yml}, relative to their original crates.
\textbf{\mbox{Implication:}} Users interacting with \tool{} may perceive these crates as risky.
}

\subsection{RQ\texorpdfstring{$_{\textcolor{red}{3}}$}{3}: How does \tool{} output fare on known vulnerable crates?}
\label{sec:rq3}

To address RQ3, we compare \tool{}'s severity labels with those from RustSec~\cite{rustsec}, a well-known database of security advisories for Rust crates. RustSec classifies crates by severity levels: critical, high, medium, low, and informational. We analyzed 592 unique crate-version pairs listed in RustSec.\footnote{This number is as of March 25, 2025.}
For this experiment, all negative assumptions related to RustSec were removed to
determine what labels would have been assigned to vulnerable crates prior to them being known to be vulnerable.

\Cref{fig:rustsec_distribution} illustrates the distribution of severity labels for crates listed on RustSec. The x-axis shows the severity labels assigned by RustSec, and each bar is split according to the labels assigned by \tool{}.
While some crates receive a higher severity label, most crates are classified as \texttt{SAFE} or \texttt{LOW\_SEVERITY}. For example, these include popular crates (hence, more likely to be marked \texttt{SAFE}) that are listed in RustSec with an \texttt{INFO} label due to being marked as unmaintained. \Cref{fig:rustsec_percentiles} confirms that this behavior is generally driven by the popularity of RustSec-listed crates: approximately 75\% of RustSec crates belong to the top 10\% of \texttt{crates.io} by downloads, and only around 15\% of RustSec crates are found in the bottom 80\%.

\begin{figure}[!t]
  \centering

  \begin{subfigure}[t]{\columnwidth}
    \centering
    \includegraphics[width=\columnwidth]{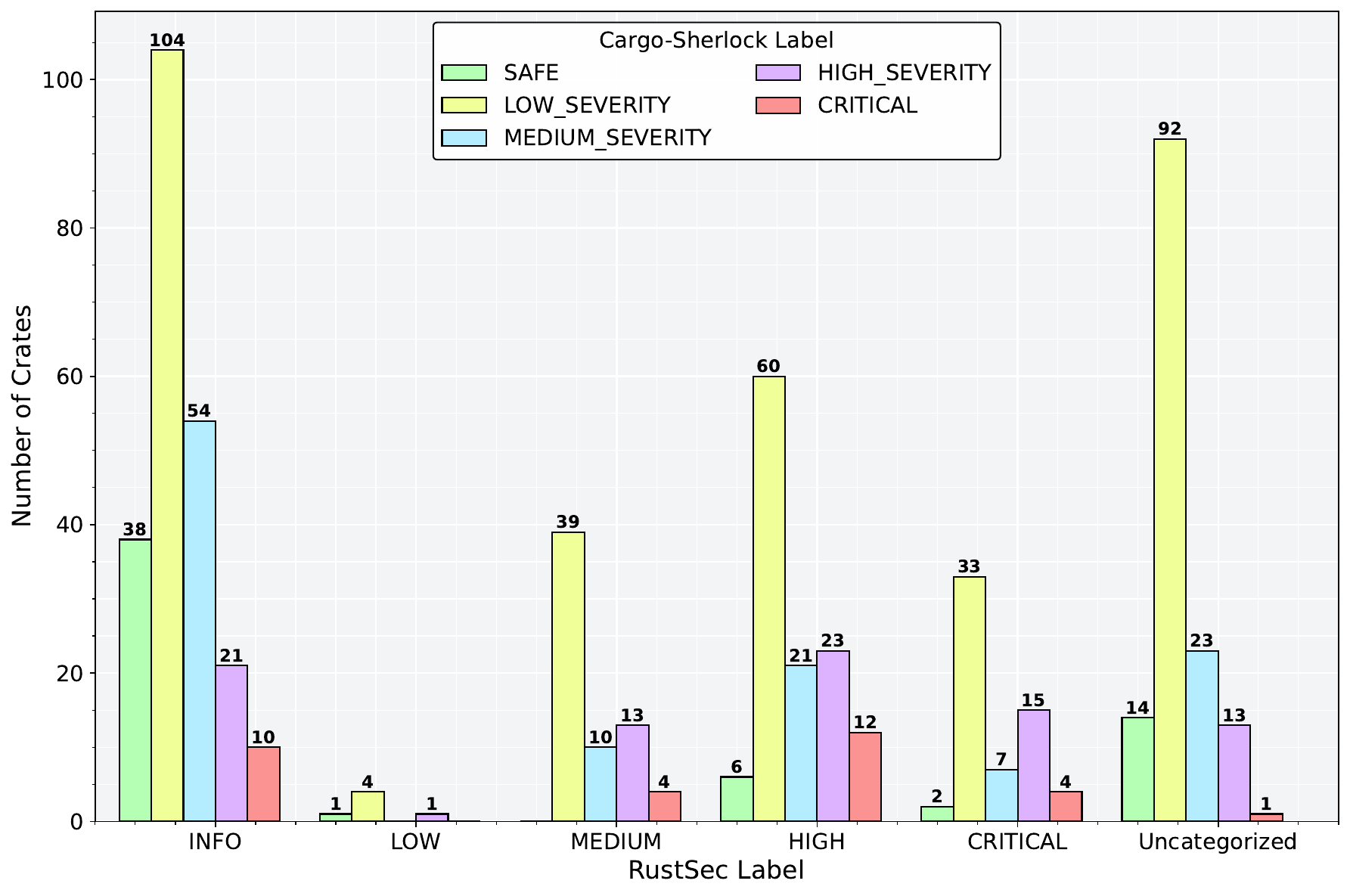}
    \caption{Distribution of \tool{} vs. RustSec severity labels.
    }
    \Description{A bar chart showing the distribution of \tool{} versus RustSec severity labels on the crates listed in RustSec.}
    \label{fig:rustsec_distribution}
  \end{subfigure}

  \medskip

  \begin{subfigure}[t]{\columnwidth}
    \centering
    \includegraphics[width=\columnwidth]{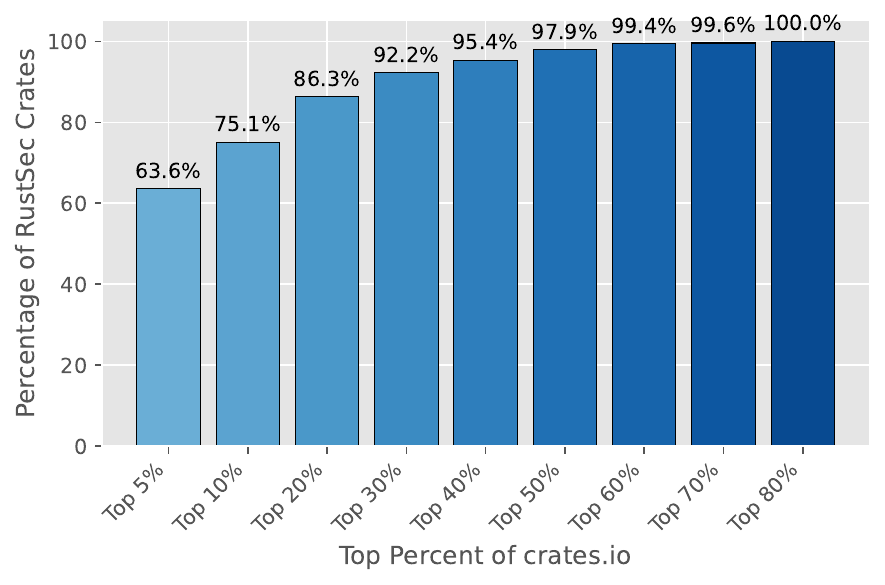}
    \caption{Distribution of crates in RustSec by downloads.
    }
    \Description{A bar chart showing the percentile distribution of crates in RustSec by downloads on crates.io.}
    \label{fig:rustsec_percentiles}
  \end{subfigure}

  \caption{\Cref{sec:rq3} experimental results.}
  \Description{RQ3 results plots.}
  \label{fig:rustsec_comparison}
\end{figure}

\findingsbox{RQ$_{\textcolor{red}{3}}$}{ Crates that are marked \texttt{SAFE} by Cargo Sherlock are not necessarily free from vulnerabilities, especially in cases with a high number of downloads.
\textbf{Implication:} Users interacting with \tool{} might perceive some vulnerable crates as safe.}

\subsection{RQ\texorpdfstring{$_{\textcolor{red}{4}}$}{4}: How does the performance of \texorpdfstring{\Cref{alg:mintrust,alg:horn-mintrust}}{Algorithms 1 and 2} scale with the number of dependencies?}
\label{sec:rq4}

To address RQ4, we sampled 1,000 random crates from \texttt{crates.io}. We excluded 47 crates from our evaluation because all of their versions had been yanked, leaving 953 crates.
To isolate algorithmic running time, we excluded data collection time from all measurements. The cumulative distribution function (CDF) plot in \Cref{fig:cdf} shows the runtime for both algorithms, using a timeout of 600 seconds. \Cref{alg:mintrust} successfully evaluates 298 crates and timed out on 500 crates. It crashes on 155 crates, out of which 150 crashes are due to the system's maximum recursion depth being exceeded. The remaining 5 crashes are shared by both algorithms and are due to parsing failures. \Cref{alg:horn-mintrust} successfully evaluated 900 crates and only timed out on 48 crates.

\begin{figure}[!t]
  \centering

  \begin{subfigure}[t]{\columnwidth}
    \centering
    \includegraphics[width=\columnwidth]{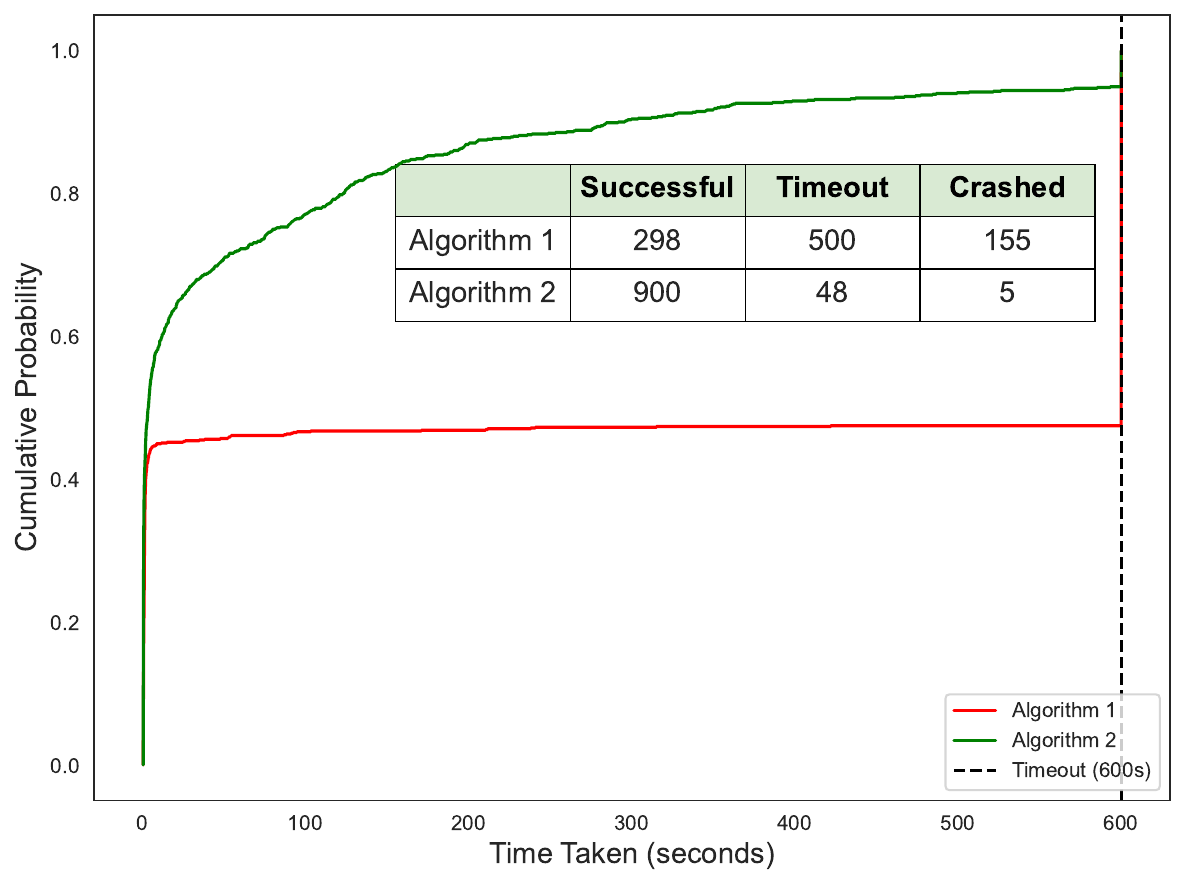}
    \caption{CDF of running time for \Cref{alg:mintrust} vs. \Cref{alg:horn-mintrust}.}
    \Description{A cumulative distribution function (CDF) plot showing the running time comparison between \Cref{alg:mintrust} and \Cref{alg:horn-mintrust} on 1000 random crates.}
    \label{fig:cdf}
  \end{subfigure}

  \medskip

  \begin{subfigure}[t]{\columnwidth}
    \centering
    \includegraphics[width=\columnwidth]{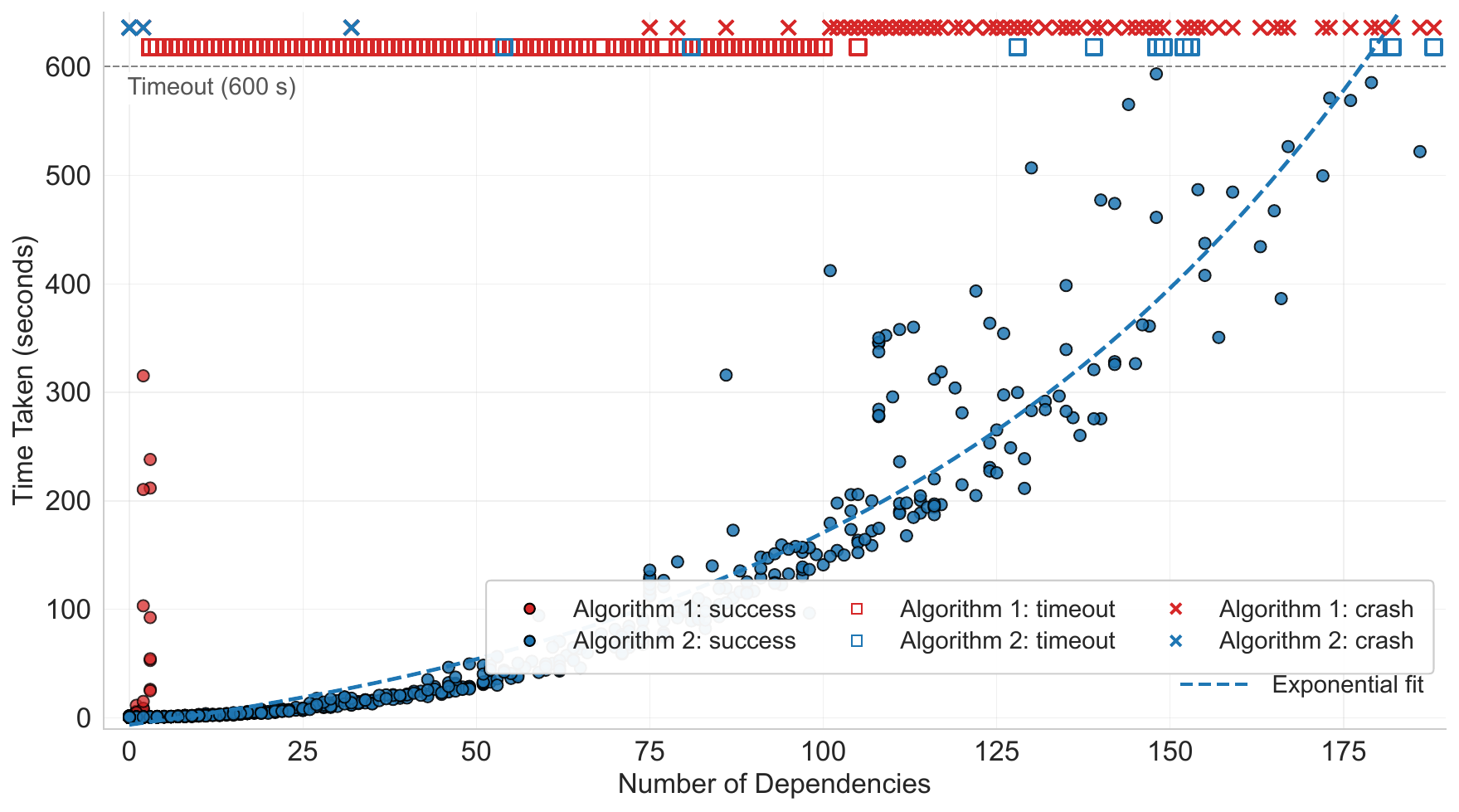}
    \caption{Time taken vs.\ number of dependencies.
    \Description{A scatter plot showing the time taken by \Cref{alg:mintrust} and \Cref{alg:horn-mintrust} relative to the number of dependencies in 1000 random crates.}
    }
    \label{fig:scatter}
  \end{subfigure}
  \caption{\Cref{sec:rq4} scalability results.}
  \Description{RQ4 results plots.}
  \label{fig:time_comparison}

  \bigskip

\end{figure}

The scatter plot (\Cref{fig:scatter}) shows that \Cref{alg:mintrust} rapidly becomes infeasible as the number of dependencies increases, with most evaluations timing out beyond approximately 10 dependencies. Conversely, \Cref{alg:horn-mintrust} scales better, successfully evaluating crates with up to 150 dependencies. A regression fit suggests that the scaling of the Horn clause algorithm is exponential. This is expected, as \code{unfold} in \Cref{alg:horn-mintrust} generates an exponential DNF formula in the worst case.

\findingsbox{RQ$_{\textcolor{red}{4}}$}{\Cref{alg:mintrust} becomes infeasible after 10 dependencies, while \Cref{alg:horn-mintrust} scales to at least 150 dependencies. \textbf{Implication:} Users can run \tool{} on crates with a large number of dependencies and receive results within 10 minutes.}

%% file: 6-rw.tex
\section{Related Work}
\label{sec:rw}

\paragraph*{Supply chain security:}
Software supply chain attacks~\cite{ohm2020backstabber,wang2021feasibility,ohm2023sok} such as the SolarWinds attack~\cite{martinez2021software} and XZ Utils backdoor~\cite{helpnetsecurityBewareBackdoor,openwallOsssecurityBackdoor} have received attention from researchers and industry.
Tools like Crev~\cite{crev} and JFrog~\cite{jfrogClosingSupply} provide support for auditing dependencies and tracking associated metadata.
Software bill of materials (SBOM)~\cite{muiri2019framing,hendrick2022state,xia2023empirical,zahan2023software}
is a general set of techniques for providing an overview of software contents, which can be used for human understanding;
these are similar to the relational model used in our work.
Tools have also been built for specific languages such as JavaScript~\cite{zimmermann2019small,cogo2019empirical,zahan2022weak}
and Go~\cite{google-capslock,google-capslock-blog}.
Especially relevant to our work, a few techniques in this space have used formal methods:
Koscinski and Mirakhorli~\cite{formally-modeled-cwes} propose a technique to formally model Common Weakness Enumerations (CWEs) using an encoding in the Alloy specification language.
Software configuration management (SCM)~\cite{zeller1997unified} is also related; it aims to identify a configuration of dependencies that satisfies formal constraints.
Our work can be thought of as a similar configuration problem to optimize for trust cost,
and selecting among multiple possible configurations (different versions of different software packages) would be an interesting direction for future work.

\paragraph*{Safety and security for Rust code:}
RustSec~\cite{rustsec} and Lib.rs~\cite{lib-rs} provide helpful databases of known vulnerabilities and libraries.
Cargo Vet~\cite{mozilla2024cargovet}, Cargo Audit~\cite{cargo-audit}, Cargo Crev~\cite{cargo-crev}, and Cargo Scan~\cite{cargo-scan,cargo-scan-paper} (which inspired the side effects metrics used in our tool)
are tools to assist with human auditing for Rust code.
Formal verification for Rust includes work on foundations --- like RustBelt~\cite{rustbelt}, Oxide~\cite{oxide}, BoCa~\cite{wagner2025linearity}, Stacked/Tree Borrows~\cite{stacked-borrows,tree-borrows}, and others~\cite{grannan2025place} --- and tools such as Verus~\cite{verus}, Prusti~\cite{prusti}, Kani~\cite{kani}, and Flux~\cite{flux}.
Program analysis tools like MIRChecker~\cite{li2021mirchecker}, Miri~\cite{miri}, and Rudra~\cite{bae2021rudra} detect memory safety and undefined behaviors in Rust code, and fuzzing tools such as SyRust~\cite{takashima2021syrust} and RULF~\cite{jiang2021rulf} use automated testing to identify bugs.
Sandboxing tools (e.g.,~\cite{trust,sandcrust,pkru-safe,xrust,fidelius-charm,galeed,gulmez2023friend}) typically aim to protect Rust code from its \emph{unsafe} subset.
However, even the \emph{safe} subset of Rust can be exploited, and this is sometimes missed by prior tools~\cite{hassnain2024counterexamples,cve-rs,cargo-scan-paper},
which motivates work incorporating human factors.

\paragraph*{Trust in software engineering:}
Hamer et al.~\cite{hamer2025trusting} explored the impact of contributor reputation and other indicator metrics in ecosystems like Rust and GitHub. Boughton et al.~\cite{boughton2024decomposing} proposed a framework to decompose and quantify trust within open-source software. We used these studies to motivate the specific assumptions that we support in our minimum trust cost model.
Devanbu et al.~\cite{devanbu1998techniques} argue for a combination of trusted hardware, key management, and verification techniques to elevate trust. Toone et al.~\cite{toone2003trust} introduce trust mediation, a framework for managing complex distributed trust relationships. Outside of computer science, models of trust have been studied in other settings (e.g.,~\cite{mayer1995integrative}).

%% file: main.bbl

\begin{thebibliography}{72}


\ifx \showCODEN    \undefined \def \showCODEN     #1{\unskip}     \fi
\ifx \showDOI      \undefined \def \showDOI       #1{#1}\fi
\ifx \showISBNx    \undefined \def \showISBNx     #1{\unskip}     \fi
\ifx \showISBNxiii \undefined \def \showISBNxiii  #1{\unskip}     \fi
\ifx \showISSN     \undefined \def \showISSN      #1{\unskip}     \fi
\ifx \showLCCN     \undefined \def \showLCCN      #1{\unskip}     \fi
\ifx \shownote     \undefined \def \shownote      #1{#1}          \fi
\ifx \showarticletitle \undefined \def \showarticletitle #1{#1}   \fi
\ifx \showURL      \undefined \def \showURL       {\relax}        \fi
\providecommand\bibfield[2]{#2}
\providecommand\bibinfo[2]{#2}
\providecommand\natexlab[1]{#1}
\providecommand\showeprint[2][]{arXiv:#2}

\bibitem[Almohri and Evans(2018)]%
        {fidelius-charm}
\bibfield{author}{\bibinfo{person}{Hussain M.~J. Almohri} {and}
  \bibinfo{person}{David Evans}.} \bibinfo{year}{2018}\natexlab{}.
\newblock \showarticletitle{{F}idelius {C}harm: Isolating Unsafe {R}ust Code}.
  In \bibinfo{booktitle}{\emph{Proceedings of the Eighth ACM Conference on Data
  and Application Security and Privacy}} (Tempe, AZ, USA)
  \emph{(\bibinfo{series}{CODASPY '18})}. \bibinfo{publisher}{Association for
  Computing Machinery}, \bibinfo{address}{New York, NY, USA},
  \bibinfo{pages}{248--255}.
\newblock
\showISBNx{9781450356329}
\urldef\tempurl%
\url{https://doi.org/10.1145/3176258.3176330}
\showDOI{\tempurl}


\bibitem[Astrauskas et~al\mbox{.}(2019)]%
        {prusti}
\bibfield{author}{\bibinfo{person}{Vytautas Astrauskas}, \bibinfo{person}{Peter
  M{\"u}ller}, \bibinfo{person}{Federico Poli}, {and}
  \bibinfo{person}{Alexander~J Summers}.} \bibinfo{year}{2019}\natexlab{}.
\newblock \showarticletitle{Leveraging {R}ust types for modular specification
  and verification}.
\newblock \bibinfo{journal}{\emph{Proceedings of the ACM on Programming
  Languages}} \bibinfo{volume}{3}, \bibinfo{number}{OOPSLA}
  (\bibinfo{year}{2019}), \bibinfo{pages}{1--30}.
\newblock


\bibitem[Bae et~al\mbox{.}(2021)]%
        {bae2021rudra}
\bibfield{author}{\bibinfo{person}{Yechan Bae}, \bibinfo{person}{Youngsuk Kim},
  \bibinfo{person}{Ammar Askar}, \bibinfo{person}{Jungwon Lim}, {and}
  \bibinfo{person}{Taesoo Kim}.} \bibinfo{year}{2021}\natexlab{}.
\newblock \showarticletitle{{R}udra: Finding Memory Safety Bugs in {R}ust at
  the Ecosystem Scale}. In \bibinfo{booktitle}{\emph{Proceedings of the ACM
  SIGOPS 28th Symposium on Operating Systems Principles}} (Virtual Event,
  Germany) \emph{(\bibinfo{series}{SOSP '21})}. \bibinfo{publisher}{Association
  for Computing Machinery}, \bibinfo{address}{New York, NY, USA},
  \bibinfo{pages}{84--99}.
\newblock
\showISBNx{9781450387095}
\urldef\tempurl%
\url{https://doi.org/10.1145/3477132.3483570}
\showDOI{\tempurl}


\bibitem[Bang et~al\mbox{.}(2023)]%
        {trust}
\bibfield{author}{\bibinfo{person}{Inyoung Bang}, \bibinfo{person}{Martin
  Kayondo}, \bibinfo{person}{HyunGon Moon}, {and} \bibinfo{person}{Yunheung
  Paek}.} \bibinfo{year}{2023}\natexlab{}.
\newblock \showarticletitle{{TRust}: A Compilation Framework for In-process
  Isolation to Protect Safe {R}ust against Untrusted Code}. In
  \bibinfo{booktitle}{\emph{32nd USENIX Security Symposium (USENIX Security
  23)}}. \bibinfo{publisher}{USENIX Association}, \bibinfo{address}{Anaheim,
  CA}, \bibinfo{pages}{6947--6964}.
\newblock
\showISBNx{978-1-939133-37-3}
\urldef\tempurl%
\url{https://www.usenix.org/conference/usenixsecurity23/presentation/bang}
\showURL{%
\tempurl}


\bibitem[Bistarelli et~al\mbox{.}(2008)]%
        {bistarelli2008weighted}
\bibfield{author}{\bibinfo{person}{Stefamno Bistarelli}, \bibinfo{person}{Fabio
  Martinelli}, {and} \bibinfo{person}{Francesco Santini}.}
  \bibinfo{year}{2008}\natexlab{}.
\newblock \showarticletitle{Weighted Datalog and Levels of Trust}. In
  \bibinfo{booktitle}{\emph{Proceedings of the 2008 Third International
  Conference on Availability, Reliability and Security}}
  \emph{(\bibinfo{series}{ARES '08})}. \bibinfo{publisher}{IEEE Computer
  Society}, \bibinfo{address}{USA}, \bibinfo{pages}{1128--1134}.
\newblock
\showISBNx{9780769531021}
\urldef\tempurl%
\url{https://doi.org/10.1109/ARES.2008.197}
\showDOI{\tempurl}


\bibitem[Boughton et~al\mbox{.}(2024)]%
        {boughton2024decomposing}
\bibfield{author}{\bibinfo{person}{Lina Boughton}, \bibinfo{person}{Courtney
  Miller}, \bibinfo{person}{Yasemin Acar}, \bibinfo{person}{Dominik Wermke},
  {and} \bibinfo{person}{Christian K\"{a}stner}.}
  \bibinfo{year}{2024}\natexlab{}.
\newblock \showarticletitle{Decomposing and Measuring Trust in Open-Source
  Software Supply Chains}. In \bibinfo{booktitle}{\emph{Proceedings of the 2024
  ACM/IEEE 44th International Conference on Software Engineering: New Ideas and
  Emerging Results}} (Lisbon, Portugal)
  \emph{(\bibinfo{series}{ICSE-NIER'24})}. \bibinfo{publisher}{Association for
  Computing Machinery}, \bibinfo{address}{New York, NY, USA},
  \bibinfo{pages}{57--61}.
\newblock
\showISBNx{9798400705007}
\urldef\tempurl%
\url{https://doi.org/10.1145/3639476.3639775}
\showDOI{\tempurl}


\bibitem[Chin(2023)]%
        {jfrogClosingSupply}
\bibfield{author}{\bibinfo{person}{Stephen Chin}.}
  \bibinfo{year}{2023}\natexlab{}.
\newblock \bibinfo{title}{{C}losing the {S}upply {C}hain {S}ecurity {L}oop with
  {R}ust @ {R}ust {N}ation {U}{K} | {J}{F}rog}.
\newblock
  \bibinfo{howpublished}{\url{https://jfrog.com/community/rust/closing-the-supply-chain-security-loop-with-rust-and-pyrsia/}}.
\newblock
\newblock
\shownote{[Accessed 2023-12]}.


\bibitem[Cogo et~al\mbox{.}(2019)]%
        {cogo2019empirical}
\bibfield{author}{\bibinfo{person}{Filipe~Roseiro Cogo},
  \bibinfo{person}{Gustavo~A Oliva}, {and} \bibinfo{person}{Ahmed~E Hassan}.}
  \bibinfo{year}{2019}\natexlab{}.
\newblock \showarticletitle{An empirical study of dependency downgrades in the
  {npm} ecosystem}.
\newblock \bibinfo{journal}{\emph{IEEE Transactions on Software Engineering}}
  \bibinfo{volume}{47}, \bibinfo{number}{11} (\bibinfo{year}{2019}),
  \bibinfo{pages}{2457--2470}.
\newblock


\bibitem[de~Moura and Bj{\o}rner(2008)]%
        {z3}
\bibfield{author}{\bibinfo{person}{Leonardo de Moura} {and}
  \bibinfo{person}{Nikolaj Bj{\o}rner}.} \bibinfo{year}{2008}\natexlab{}.
\newblock \showarticletitle{{Z3}: An Efficient {SMT} Solver}. In
  \bibinfo{booktitle}{\emph{Tools and Algorithms for the Construction and
  Analysis of Systems}}, \bibfield{editor}{\bibinfo{person}{C.~R. Ramakrishnan}
  {and} \bibinfo{person}{Jakob Rehof}} (Eds.). \bibinfo{publisher}{Springer
  Berlin Heidelberg}, \bibinfo{address}{Berlin, Heidelberg},
  \bibinfo{pages}{337--340}.
\newblock
\showISBNx{978-3-540-78800-3}


\bibitem[Devanbu et~al\mbox{.}(1998)]%
        {devanbu1998techniques}
\bibfield{author}{\bibinfo{person}{Premkumar~T. Devanbu},
  \bibinfo{person}{Philip W-L Fong}, {and} \bibinfo{person}{Stuart~G.
  Stubblebine}.} \bibinfo{year}{1998}\natexlab{}.
\newblock \showarticletitle{Techniques for trusted software engineering}. In
  \bibinfo{booktitle}{\emph{Proceedings of the 20th International Conference on
  Software Engineering}} (Kyoto, Japan) \emph{(\bibinfo{series}{ICSE '98})}.
  \bibinfo{publisher}{IEEE Computer Society}, \bibinfo{address}{USA},
  \bibinfo{pages}{126--135}.
\newblock
\showISBNx{0818683686}


\bibitem[development team(2024a)]%
        {cargo-crev}
\bibfield{author}{\bibinfo{person}{Crev development team}.}
  \bibinfo{year}{2024}\natexlab{a}.
\newblock \bibinfo{title}{{G}it{H}ub -- crev-dev/cargo-crev: {A}
  cryptographically verifiable code review system for the cargo ({R}ust)
  package manager.}
\newblock \bibinfo{howpublished}{\url{https://github.com/crev-dev/cargo-crev}}.
\newblock
\newblock
\shownote{[Accessed 14-11-2024]}.


\bibitem[development team(2024b)]%
        {crev}
\bibfield{author}{\bibinfo{person}{Crev development team}.}
  \bibinfo{year}{2024}\natexlab{b}.
\newblock \bibinfo{title}{{G}it{H}ub -- crev-dev/crev: {Crev} - {C}ode {REV}iew
  system that we desperately need}.
\newblock \bibinfo{howpublished}{\url{https://github.com/crev-dev/crev}}.
\newblock
\newblock
\shownote{[Accessed 14-11-2024]}.


\bibitem[Dimoulas and Felleisen(2025)]%
        {dimoulas2025rational}
\bibfield{author}{\bibinfo{person}{Christos Dimoulas} {and}
  \bibinfo{person}{Matthias Felleisen}.} \bibinfo{year}{2025}\natexlab{}.
\newblock \showarticletitle{The Rational Programmer: Investigating Programming
  Language Pragmatics}.
\newblock \bibinfo{journal}{\emph{Commun. ACM}} \bibinfo{volume}{68},
  \bibinfo{number}{7} (\bibinfo{year}{2025}), \bibinfo{pages}{120--130}.
\newblock


\bibitem[Freund(2024)]%
        {openwallOsssecurityBackdoor}
\bibfield{author}{\bibinfo{person}{Andres Freund}.}
  \bibinfo{year}{2024}\natexlab{}.
\newblock \bibinfo{title}{(oss-security) backdoor in upstream xz/liblzma
  leading to ssh server compromise}.
\newblock
  \bibinfo{howpublished}{\url{https://www.openwall.com/lists/oss-security/2024/03/29/4}}.
\newblock
\newblock
\shownote{[Accessed 05-11-2025]}.


\bibitem[Ghaffarian and Shahriari(2017)]%
        {ghaffarian2017software}
\bibfield{author}{\bibinfo{person}{Seyed~Mohammad Ghaffarian} {and}
  \bibinfo{person}{Hamid~Reza Shahriari}.} \bibinfo{year}{2017}\natexlab{}.
\newblock \showarticletitle{Software vulnerability analysis and discovery using
  machine-learning and data-mining techniques: A survey}.
\newblock \bibinfo{journal}{\emph{ACM computing surveys (CSUR)}}
  \bibinfo{volume}{50}, \bibinfo{number}{4} (\bibinfo{year}{2017}),
  \bibinfo{pages}{1--36}.
\newblock


\bibitem[Google(2023a)]%
        {google-capslock-blog}
\bibfield{author}{\bibinfo{person}{Google}.} \bibinfo{year}{2023}\natexlab{a}.
\newblock \bibinfo{title}{{C}apslock: {W}hat is your code really capable of?}
\newblock
  \bibinfo{howpublished}{\url{https://security.googleblog.com/2023/09/capslock-what-is-your-code-really.html}}.
\newblock
\newblock
\shownote{[Accessed 14-11-2024]}.


\bibitem[Google(2023b)]%
        {google-capslock}
\bibfield{author}{\bibinfo{person}{Google}.} \bibinfo{year}{2023}\natexlab{b}.
\newblock \bibinfo{title}{{G}it{H}ub -- Google Capslock}.
\newblock \bibinfo{howpublished}{\url{https://github.com/google/capslock}}.
\newblock
\newblock
\shownote{[Accessed 14-11-2024]}.


\bibitem[Grannan et~al\mbox{.}(2025)]%
        {grannan2025place}
\bibfield{author}{\bibinfo{person}{Zachary Grannan}, \bibinfo{person}{Aurel
  B{\'\i}l{\`y}}, \bibinfo{person}{Jon{\'a}{\v{s}} Fiala},
  \bibinfo{person}{Jasper Geer}, \bibinfo{person}{Markus de Medeiros},
  \bibinfo{person}{Peter M{\"u}ller}, {and} \bibinfo{person}{Alexander~J
  Summers}.} \bibinfo{year}{2025}\natexlab{}.
\newblock \showarticletitle{Place Capability Graphs: A General-Purpose Model of
  {R}ust's Ownership and Borrowing Guarantees}.
\newblock \bibinfo{journal}{\emph{Proceedings of the ACM on Programming
  Languages}} \bibinfo{volume}{9}, \bibinfo{number}{OOPSLA2}
  (\bibinfo{year}{2025}), \bibinfo{pages}{2002--2029}.
\newblock


\bibitem[Green and Tannen(2017)]%
        {green2017semiring}
\bibfield{author}{\bibinfo{person}{Todd~J. Green} {and} \bibinfo{person}{Val
  Tannen}.} \bibinfo{year}{2017}\natexlab{}.
\newblock \showarticletitle{The Semiring Framework for Database Provenance}. In
  \bibinfo{booktitle}{\emph{Proceedings of the 36th ACM SIGMOD-SIGACT-SIGAI
  Symposium on Principles of Database Systems}} (Chicago, Illinois, USA)
  \emph{(\bibinfo{series}{PODS '17})}. \bibinfo{publisher}{Association for
  Computing Machinery}, \bibinfo{address}{New York, NY, USA},
  \bibinfo{pages}{93--99}.
\newblock
\showISBNx{9781450341981}
\urldef\tempurl%
\url{https://doi.org/10.1145/3034786.3056125}
\showDOI{\tempurl}


\bibitem[Gulmez et~al\mbox{.}(2023)]%
        {gulmez2023friend}
\bibfield{author}{\bibinfo{person}{Merve Gulmez}, \bibinfo{person}{Thomas
  Nyman}, \bibinfo{person}{Christoph Baumann}, {and}
  \bibinfo{person}{Jan~Tobias Muhlberg}.} \bibinfo{year}{2023}\natexlab{}.
\newblock \showarticletitle{{Friend or Foe Inside? Exploring In-Process
  Isolation to Maintain Memory Safety for Unsafe {R}ust}}. In
  \bibinfo{booktitle}{\emph{2023 IEEE Secure Development Conference (SecDev)}}.
  \bibinfo{publisher}{IEEE Computer Society}, \bibinfo{address}{Los Alamitos,
  CA, USA}, \bibinfo{pages}{54--66}.
\newblock
\urldef\tempurl%
\url{https://doi.org/10.1109/SecDev56634.2023.00020}
\showDOI{\tempurl}


\bibitem[Hamer et~al\mbox{.}(2025)]%
        {hamer2025trusting}
\bibfield{author}{\bibinfo{person}{Sivana Hamer}, \bibinfo{person}{Nasif
  Imtiaz}, \bibinfo{person}{Mahzabin Tamanna}, \bibinfo{person}{Preya
  Shabrina}, {and} \bibinfo{person}{Laurie Williams}.}
  \bibinfo{year}{2025}\natexlab{}.
\newblock \showarticletitle{Trusting Code in the Wild: Exploring Contributor
  Reputation Measures to Review Dependencies in the {R}ust Ecosystem}.
\newblock \bibinfo{journal}{\emph{IEEE Trans. Softw. Eng.}}
  \bibinfo{volume}{51}, \bibinfo{number}{4} (\bibinfo{date}{April}
  \bibinfo{year}{2025}), \bibinfo{pages}{1319--1333}.
\newblock
\showISSN{0098-5589}
\urldef\tempurl%
\url{https://doi.org/10.1109/TSE.2025.3551664}
\showDOI{\tempurl}


\bibitem[Hassnain and Stanford(2024)]%
        {hassnain2024counterexamples}
\bibfield{author}{\bibinfo{person}{Muhammad Hassnain} {and}
  \bibinfo{person}{Caleb Stanford}.} \bibinfo{year}{2024}\natexlab{}.
\newblock \showarticletitle{Counterexamples in Safe {R}ust}. In
  \bibinfo{booktitle}{\emph{Proceedings of the 39th IEEE/ACM International
  Conference on Automated Software Engineering Workshops}} (Sacramento, CA,
  USA) \emph{(\bibinfo{series}{ASEW '24})}. \bibinfo{publisher}{Association for
  Computing Machinery}, \bibinfo{address}{New York, NY, USA},
  \bibinfo{pages}{128--135}.
\newblock
\showISBNx{9798400712494}
\urldef\tempurl%
\url{https://doi.org/10.1145/3691621.3694943}
\showDOI{\tempurl}


\bibitem[Hendrick and Zemlin(2022)]%
        {hendrick2022state}
\bibfield{author}{\bibinfo{person}{Stephen Hendrick} {and} \bibinfo{person}{J
  Zemlin}.} \bibinfo{year}{2022}\natexlab{}.
\newblock \bibinfo{booktitle}{\emph{The state of software bill of materials
  ({SBOM}) and cybersecurity readiness}}.
\newblock \bibinfo{type}{{T}echnical {R}eport}. \bibinfo{institution}{Retrieved
  2023-08-18 from https://www. linuxfoundation. org/research/the~…}.
\newblock


\bibitem[Heras et~al\mbox{.}(2008)]%
        {heras2008minimaxsat}
\bibfield{author}{\bibinfo{person}{Federico Heras}, \bibinfo{person}{Javier
  Larrosa}, {and} \bibinfo{person}{Albert Oliveras}.}
  \bibinfo{year}{2008}\natexlab{}.
\newblock \showarticletitle{{MiniMaxSAT}: An efficient weighted Max-{SAT}
  solver}.
\newblock \bibinfo{journal}{\emph{Journal of Artificial Intelligence Research}}
   \bibinfo{volume}{31} (\bibinfo{year}{2008}), \bibinfo{pages}{1--32}.
\newblock


\bibitem[Jiang et~al\mbox{.}(2022)]%
        {jiang2021rulf}
\bibfield{author}{\bibinfo{person}{Jianfeng Jiang}, \bibinfo{person}{Hui Xu},
  {and} \bibinfo{person}{Yangfan Zhou}.} \bibinfo{year}{2022}\natexlab{}.
\newblock \showarticletitle{{RULF}: {R}ust library fuzzing via API dependency
  graph traversal}. In \bibinfo{booktitle}{\emph{Proceedings of the 36th
  IEEE/ACM International Conference on Automated Software Engineering}}
  (Melbourne, Australia) \emph{(\bibinfo{series}{ASE '21})}.
  \bibinfo{publisher}{IEEE Press}, \bibinfo{address}{Los Alamitos, CA, USA},
  \bibinfo{pages}{581--592}.
\newblock
\showISBNx{9781665403375}
\urldef\tempurl%
\url{https://doi.org/10.1109/ASE51524.2021.9678813}
\showDOI{\tempurl}


\bibitem[Jung et~al\mbox{.}(2019)]%
        {stacked-borrows}
\bibfield{author}{\bibinfo{person}{Ralf Jung}, \bibinfo{person}{Hoang-Hai
  Dang}, \bibinfo{person}{Jeehoon Kang}, {and} \bibinfo{person}{Derek Dreyer}.}
  \bibinfo{year}{2019}\natexlab{}.
\newblock \showarticletitle{Stacked Borrows: An Aliasing Model for {R}ust}.
\newblock \bibinfo{journal}{\emph{Proc. ACM Program. Lang.}}
  \bibinfo{volume}{4}, \bibinfo{number}{POPL}, Article \bibinfo{articleno}{41}
  (\bibinfo{date}{dec} \bibinfo{year}{2019}), \bibinfo{numpages}{32}~pages.
\newblock
\urldef\tempurl%
\url{https://doi.org/10.1145/3371109}
\showDOI{\tempurl}


\bibitem[Jung et~al\mbox{.}(2017)]%
        {rustbelt}
\bibfield{author}{\bibinfo{person}{Ralf Jung}, \bibinfo{person}{Jacques-Henri
  Jourdan}, \bibinfo{person}{Robbert Krebbers}, {and} \bibinfo{person}{Derek
  Dreyer}.} \bibinfo{year}{2017}\natexlab{}.
\newblock \showarticletitle{{RustBelt}: Securing the Foundations of the {R}ust
  Programming Language}.
\newblock \bibinfo{journal}{\emph{Proc. ACM Program. Lang.}}
  \bibinfo{volume}{2}, \bibinfo{number}{POPL}, Article \bibinfo{articleno}{66}
  (\bibinfo{date}{dec} \bibinfo{year}{2017}), \bibinfo{numpages}{34}~pages.
\newblock
\urldef\tempurl%
\url{https://doi.org/10.1145/3158154}
\showDOI{\tempurl}


\bibitem[Karp(1972)]%
        {karp1972reducibility}
\bibfield{author}{\bibinfo{person}{Richard Karp}.}
  \bibinfo{year}{1972}\natexlab{}.
\newblock \showarticletitle{Reducibility Among Combinatorial Problems}.
\newblock \bibinfo{journal}{\emph{Complexity of Computer Computations}}
  \bibinfo{volume}{40} (\bibinfo{date}{01} \bibinfo{year}{1972}),
  \bibinfo{pages}{85--103}.
\newblock
\showISBNx{978-3-540-68274-5}
\urldef\tempurl%
\url{https://doi.org/10.1007/978-3-540-68279-0_8}
\showDOI{\tempurl}


\bibitem[Kirth et~al\mbox{.}(2022)]%
        {pkru-safe}
\bibfield{author}{\bibinfo{person}{Paul Kirth}, \bibinfo{person}{Mitchel
  Dickerson}, \bibinfo{person}{Stephen Crane}, \bibinfo{person}{Per Larsen},
  \bibinfo{person}{Adrian Dabrowski}, \bibinfo{person}{David Gens},
  \bibinfo{person}{Yeoul Na}, \bibinfo{person}{Stijn Volckaert}, {and}
  \bibinfo{person}{Michael Franz}.} \bibinfo{year}{2022}\natexlab{}.
\newblock \showarticletitle{{PKRU}-Safe: Automatically Locking down the Heap
  between Safe and Unsafe Languages}. In \bibinfo{booktitle}{\emph{Proceedings
  of the Seventeenth European Conference on Computer Systems}} (Rennes, France)
  \emph{(\bibinfo{series}{EuroSys '22})}. \bibinfo{publisher}{Association for
  Computing Machinery}, \bibinfo{address}{New York, NY, USA},
  \bibinfo{pages}{132--148}.
\newblock
\showISBNx{9781450391627}
\urldef\tempurl%
\url{https://doi.org/10.1145/3492321.3519582}
\showDOI{\tempurl}


\bibitem[Koscinski and Mirakhorli(2024)]%
        {formally-modeled-cwes}
\bibfield{author}{\bibinfo{person}{Viktoria Koscinski} {and}
  \bibinfo{person}{Mehdi Mirakhorli}.} \bibinfo{year}{2024}\natexlab{}.
\newblock \showarticletitle{Formally Modeled Common Weakness Enumerations
  ({CWE}s)}. In \bibinfo{booktitle}{\emph{Proceedings of the 39th IEEE/ACM
  International Conference on Automated Software Engineering Workshops}}
  (Sacramento, CA, USA) \emph{(\bibinfo{series}{ASEW '24})}.
  \bibinfo{publisher}{Association for Computing Machinery},
  \bibinfo{address}{New York, NY, USA}, \bibinfo{pages}{88--93}.
\newblock
\showISBNx{9798400712494}
\urldef\tempurl%
\url{https://doi.org/10.1145/3691621.3694938}
\showDOI{\tempurl}


\bibitem[Kowalski(1974)]%
        {kowalski1974predicate}
\bibfield{author}{\bibinfo{person}{Robert~A Kowalski}.}
  \bibinfo{year}{1974}\natexlab{}.
\newblock \showarticletitle{Predicate Logic as a Programming Language}.
\newblock \bibinfo{journal}{\emph{IFIP Congress}}  \bibinfo{volume}{74}
  (\bibinfo{year}{1974}), \bibinfo{pages}{569--574}.
\newblock


\bibitem[Krsul(1998)]%
        {krsul1998software}
\bibfield{author}{\bibinfo{person}{Ivan~Victor Krsul}.}
  \bibinfo{year}{1998}\natexlab{}.
\newblock \bibinfo{booktitle}{\emph{Software vulnerability analysis}}.
\newblock \bibinfo{publisher}{Purdue University}, \bibinfo{address}{West
  Lafayette, IN, USA}.
\newblock


\bibitem[Lamowski et~al\mbox{.}(2017)]%
        {sandcrust}
\bibfield{author}{\bibinfo{person}{Benjamin Lamowski}, \bibinfo{person}{Carsten
  Weinhold}, \bibinfo{person}{Adam Lackorzynski}, {and}
  \bibinfo{person}{Hermann H\"{a}rtig}.} \bibinfo{year}{2017}\natexlab{}.
\newblock \showarticletitle{{S}andcrust: Automatic Sandboxing of Unsafe
  Components in {R}ust}. In \bibinfo{booktitle}{\emph{Proceedings of the 9th
  Workshop on Programming Languages and Operating Systems}} (Shanghai, China)
  \emph{(\bibinfo{series}{PLOS '17})}. \bibinfo{publisher}{Association for
  Computing Machinery}, \bibinfo{address}{New York, NY, USA},
  \bibinfo{pages}{51--57}.
\newblock
\showISBNx{9781450351539}
\urldef\tempurl%
\url{https://doi.org/10.1145/3144555.3144562}
\showDOI{\tempurl}


\bibitem[language team(2022)]%
        {verus}
\bibfield{author}{\bibinfo{person}{Verus language team}.}
  \bibinfo{year}{2022}\natexlab{}.
\newblock \bibinfo{title}{verus-lang/verus: {V}erified {R}ust for low-level
  systems code}.
\newblock \bibinfo{howpublished}{\url{https://github.com/verus-lang/verus}}.
\newblock
\newblock
\shownote{[Accessed 2023-2-03]}.


\bibitem[Lehmann et~al\mbox{.}(2023)]%
        {flux}
\bibfield{author}{\bibinfo{person}{Nico Lehmann}, \bibinfo{person}{Adam~T
  Geller}, \bibinfo{person}{Niki Vazou}, {and} \bibinfo{person}{Ranjit Jhala}.}
  \bibinfo{year}{2023}\natexlab{}.
\newblock \showarticletitle{{F}lux: Liquid types for {R}ust}.
\newblock \bibinfo{journal}{\emph{Proceedings of the ACM on Programming
  Languages}} \bibinfo{volume}{7}, \bibinfo{number}{PLDI}
  (\bibinfo{year}{2023}), \bibinfo{pages}{1533--1557}.
\newblock


\bibitem[Lesi\'nski(2024)]%
        {lib-rs}
\bibfield{author}{\bibinfo{person}{Kornel Lesi\'nski}.}
  \bibinfo{year}{2024}\natexlab{}.
\newblock \bibinfo{title}{{L}ib.rs -- home for {R}ust crates}.
\newblock \bibinfo{howpublished}{\url{https://lib.rs/}}.
\newblock
\newblock
\shownote{[Accessed 14-11-2024]}.


\bibitem[Li and Manya(2009)]%
        {li2009maxsat}
\bibfield{author}{\bibinfo{person}{Chu~Min Li} {and} \bibinfo{person}{Felip
  Manya}.} \bibinfo{year}{2009}\natexlab{}.
\newblock \showarticletitle{Max{SAT}, hard and soft constraints}.
\newblock In \bibinfo{booktitle}{\emph{Handbook of satisfiability}}.
  \bibinfo{publisher}{IOS Press}, \bibinfo{address}{Amsterdam, The
  Netherlands}, \bibinfo{pages}{613--631}.
\newblock


\bibitem[Li et~al\mbox{.}(2021)]%
        {li2021mirchecker}
\bibfield{author}{\bibinfo{person}{Zhuohua Li}, \bibinfo{person}{Jincheng
  Wang}, \bibinfo{person}{Mingshen Sun}, {and} \bibinfo{person}{John~C.S.
  Lui}.} \bibinfo{year}{2021}\natexlab{}.
\newblock \showarticletitle{{MirChecker}: Detecting Bugs in {R}ust Programs via
  Static Analysis}. In \bibinfo{booktitle}{\emph{Proceedings of the 2021 ACM
  SIGSAC Conference on Computer and Communications Security}} (Virtual Event,
  Republic of Korea) \emph{(\bibinfo{series}{CCS '21})}.
  \bibinfo{publisher}{Association for Computing Machinery},
  \bibinfo{address}{New York, NY, USA}, \bibinfo{pages}{2183--2196}.
\newblock
\showISBNx{9781450384544}
\urldef\tempurl%
\url{https://doi.org/10.1145/3460120.3484541}
\showDOI{\tempurl}


\bibitem[Liu et~al\mbox{.}(2020)]%
        {xrust}
\bibfield{author}{\bibinfo{person}{Peiming Liu}, \bibinfo{person}{Gang Zhao},
  {and} \bibinfo{person}{Jeff Huang}.} \bibinfo{year}{2020}\natexlab{}.
\newblock \showarticletitle{Securing Unsafe {R}ust Programs with {XRust}}. In
  \bibinfo{booktitle}{\emph{Proceedings of the ACM/IEEE 42nd International
  Conference on Software Engineering}} (Seoul, South Korea)
  \emph{(\bibinfo{series}{ICSE '20})}. \bibinfo{publisher}{Association for
  Computing Machinery}, \bibinfo{address}{New York, NY, USA},
  \bibinfo{pages}{234--245}.
\newblock
\showISBNx{9781450371216}
\urldef\tempurl%
\url{https://doi.org/10.1145/3377811.3380325}
\showDOI{\tempurl}


\bibitem[Mart{\'\i}nez and Dur{\'a}n(2021)]%
        {martinez2021software}
\bibfield{author}{\bibinfo{person}{Jeferson Mart{\'\i}nez} {and}
  \bibinfo{person}{Javier~M Dur{\'a}n}.} \bibinfo{year}{2021}\natexlab{}.
\newblock \showarticletitle{Software supply chain attacks, a threat to global
  cybersecurity: {S}olar{W}inds' case study}.
\newblock \bibinfo{journal}{\emph{International Journal of Safety and Security
  Engineering}} \bibinfo{volume}{11}, \bibinfo{number}{5}
  (\bibinfo{year}{2021}), \bibinfo{pages}{537--545}.
\newblock


\bibitem[Mayer et~al\mbox{.}(1995)]%
        {mayer1995integrative}
\bibfield{author}{\bibinfo{person}{Roger~C Mayer}, \bibinfo{person}{James~H
  Davis}, {and} \bibinfo{person}{F~David Schoorman}.}
  \bibinfo{year}{1995}\natexlab{}.
\newblock \showarticletitle{An integrative model of organizational trust}.
\newblock \bibinfo{journal}{\emph{Academy of management review}}
  \bibinfo{volume}{20}, \bibinfo{number}{3} (\bibinfo{year}{1995}),
  \bibinfo{pages}{709--734}.
\newblock


\bibitem[M{\o}ller and Schwartzbach(2012)]%
        {moller2012static}
\bibfield{author}{\bibinfo{person}{Anders M{\o}ller} {and}
  \bibinfo{person}{Michael~I Schwartzbach}.} \bibinfo{year}{2012}\natexlab{}.
\newblock \bibinfo{title}{Static program analysis}.
\newblock \bibinfo{howpublished}{Lecture notes}.
\newblock
\newblock
\shownote{February 2012}.


\bibitem[Mozilla(2024)]%
        {mozilla2024cargovet}
\bibfield{author}{\bibinfo{person}{Mozilla}.} \bibinfo{year}{2024}\natexlab{}.
\newblock \bibinfo{title}{{cargo-vet}: Supply-Chain Security for {R}ust}.
\newblock \bibinfo{howpublished}{\url{https://github.com/mozilla/cargo-vet}}.
\newblock
\newblock
\shownote{[Accessed 05-11-2024]}.


\bibitem[Muir{\'\i}(2019)]%
        {muiri2019framing}
\bibfield{author}{\bibinfo{person}{{\'E}amonn~{\'O} Muir{\'\i}}.}
  \bibinfo{year}{2019}\natexlab{}.
\newblock \bibinfo{title}{Framing software component transparency: Establishing
  a common software bill of material ({SBOM})}.
\newblock \bibinfo{howpublished}{National Telecommunications and Information
  Administration}.
\newblock
\newblock
\shownote{November 2019}.


\bibitem[Ohm et~al\mbox{.}(2020)]%
        {ohm2020backstabber}
\bibfield{author}{\bibinfo{person}{Marc Ohm}, \bibinfo{person}{Henrik Plate},
  \bibinfo{person}{Arnold Sykosch}, {and} \bibinfo{person}{Michael Meier}.}
  \bibinfo{year}{2020}\natexlab{}.
\newblock \showarticletitle{Backstabber's Knife Collection: A Review of Open
  Source Software Supply Chain Attacks}. In \bibinfo{booktitle}{\emph{Detection
  of Intrusions and Malware, and Vulnerability Assessment: 17th International
  Conference, DIMVA 2020, Lisbon, Portugal, June 24-26, 2020, Proceedings}}
  (Lisbon, Portugal). \bibinfo{publisher}{Springer-Verlag},
  \bibinfo{address}{Berlin, Heidelberg}, \bibinfo{pages}{23--43}.
\newblock
\showISBNx{978-3-030-52682-5}
\urldef\tempurl%
\url{https://doi.org/10.1007/978-3-030-52683-2_2}
\showDOI{\tempurl}


\bibitem[Ohm and Stuke(2023)]%
        {ohm2023sok}
\bibfield{author}{\bibinfo{person}{Marc Ohm} {and} \bibinfo{person}{Charlene
  Stuke}.} \bibinfo{year}{2023}\natexlab{}.
\newblock \showarticletitle{{SoK}: Practical Detection of Software Supply Chain
  Attacks}. In \bibinfo{booktitle}{\emph{Proceedings of the 18th International
  Conference on Availability, Reliability and Security}} (Benevento, Italy)
  \emph{(\bibinfo{series}{ARES '23})}. \bibinfo{publisher}{Association for
  Computing Machinery}, \bibinfo{address}{New York, NY, USA}, Article
  \bibinfo{articleno}{33}, \bibinfo{numpages}{11}~pages.
\newblock
\showISBNx{9798400707728}
\urldef\tempurl%
\url{https://doi.org/10.1145/3600160.3600162}
\showDOI{\tempurl}


\bibitem[Pearce(2025)]%
        {fasterlog-asyncprintln}
\bibfield{author}{\bibinfo{person}{Walter Pearce}.}
  \bibinfo{year}{2025}\natexlab{}.
\newblock \bibinfo{title}{crates.io: Malicious {crates} {faster\_log} and
  {async\_println}}.
\newblock
  \bibinfo{howpublished}{\url{https://blog.rust-lang.org/2025/09/24/crates.io-malicious-crates-fasterlog-and-asyncprintln/}}.
\newblock
\newblock
\shownote{Rust Blog; on behalf of the crates.io team. Accessed 2025-10-14}.


\bibitem[PL and PLSysSec(2026)]%
        {cargo-scan}
\bibfield{author}{\bibinfo{person}{UC~Davis PL} {and}
  \bibinfo{person}{UC~San~Diego PLSysSec}.} \bibinfo{year}{2026}\natexlab{}.
\newblock \bibinfo{title}{{C}argo {S}can: a tool for auditing {R}ust crates}.
\newblock \bibinfo{howpublished}{\url{https://github.com/PLSysSec/cargo-scan}}.
\newblock
\newblock
\shownote{[Accessed 10-09-2024]}.


\bibitem[Rivera et~al\mbox{.}(2021)]%
        {galeed}
\bibfield{author}{\bibinfo{person}{Elijah Rivera}, \bibinfo{person}{Samuel
  Mergendahl}, \bibinfo{person}{Howard Shrobe}, \bibinfo{person}{Hamed
  Okhravi}, {and} \bibinfo{person}{Nathan Burow}.}
  \bibinfo{year}{2021}\natexlab{}.
\newblock \showarticletitle{Keeping Safe {R}ust Safe with {G}aleed}. In
  \bibinfo{booktitle}{\emph{Proceedings of the 37th Annual Computer Security
  Applications Conference}} (Virtual Event, USA) \emph{(\bibinfo{series}{ACSAC
  '21})}. \bibinfo{publisher}{Association for Computing Machinery},
  \bibinfo{address}{New York, NY, USA}, \bibinfo{pages}{824--836}.
\newblock
\showISBNx{9781450385794}
\urldef\tempurl%
\url{https://doi.org/10.1145/3485832.3485903}
\showDOI{\tempurl}


\bibitem[{R}ust Secure Code Working~Group(2024a)]%
        {cargo-audit}
\bibfield{author}{\bibinfo{person}{{R}ust Secure Code Working~Group}.}
  \bibinfo{year}{2024}\natexlab{a}.
\newblock \bibinfo{title}{{G}it{H}ub -- {RustSec}: cargo audit}.
\newblock
  \bibinfo{howpublished}{\url{https://github.com/rustsec/rustsec/tree/main/cargo-audit}}.
\newblock
\newblock
\shownote{[Accessed 14-11-2024]}.


\bibitem[{R}ust Secure Code Working~Group(2024b)]%
        {rustsecRUSTSEC20220042Rustdecimal}
\bibfield{author}{\bibinfo{person}{{R}ust Secure Code Working~Group}.}
  \bibinfo{year}{2024}\natexlab{b}.
\newblock \bibinfo{title}{{R}{U}{S}{T}{S}{E}{C}-2022-0042: malicious crate
  `rustdecimal' --- {R}ust{S}ec {A}dvisory {D}atabase}.
\newblock
  \bibinfo{howpublished}{\url{https://rustsec.org/advisories/RUSTSEC-2022-0042.html}}.
\newblock
\newblock
\shownote{[Accessed 03-08-2024]}.


\bibitem[{R}ust Secure Code Working~Group(2025)]%
        {rustsec}
\bibfield{author}{\bibinfo{person}{{R}ust Secure Code Working~Group}.}
  \bibinfo{year}{2025}\natexlab{}.
\newblock \bibinfo{title}{{RustSec}: {R}ust {S}ecurity {A}dvisory {D}atabase}.
\newblock \bibinfo{howpublished}{\url{https://rustsec.org/advisories/}}.
\newblock
\newblock
\shownote{[Accessed 06-11-2024]}.


\bibitem[{R}ust team(2024)]%
        {miri}
\bibfield{author}{\bibinfo{person}{{R}ust team}.}
  \bibinfo{year}{2024}\natexlab{}.
\newblock \bibinfo{title}{{G}it{H}ub -- rust-lang/miri: {A}n interpreter for
  {R}ust's mid-level intermediate representation}.
\newblock \bibinfo{howpublished}{\url{https://github.com/rust-lang/miri}}.
\newblock
\newblock
\shownote{[Accessed 2023-12-06]}.


\bibitem[Schaefer and Umans(2002)]%
        {schaefer2001phcompleteness}
\bibfield{author}{\bibinfo{person}{Marcus Schaefer} {and}
  \bibinfo{person}{Christopher Umans}.} \bibinfo{year}{2002}\natexlab{}.
\newblock \showarticletitle{Completeness in the polynomial-time hierarchy: A
  compendium}.
\newblock \bibinfo{journal}{\emph{SIGACT news}} \bibinfo{volume}{33},
  \bibinfo{number}{3} (\bibinfo{year}{2002}), \bibinfo{pages}{32--49}.
\newblock


\bibitem[Takashima et~al\mbox{.}(2021)]%
        {takashima2021syrust}
\bibfield{author}{\bibinfo{person}{Yoshiki Takashima}, \bibinfo{person}{Ruben
  Martins}, \bibinfo{person}{Limin Jia}, {and} \bibinfo{person}{Corina~S.
  P\u{a}s\u{a}reanu}.} \bibinfo{year}{2021}\natexlab{}.
\newblock \showarticletitle{{SyRust}: automatic testing of {R}ust libraries
  with semantic-aware program synthesis}. In
  \bibinfo{booktitle}{\emph{Proceedings of the 42nd ACM SIGPLAN International
  Conference on Programming Language Design and Implementation}} (Virtual,
  Canada) \emph{(\bibinfo{series}{PLDI 2021})}. \bibinfo{publisher}{Association
  for Computing Machinery}, \bibinfo{address}{New York, NY, USA},
  \bibinfo{pages}{899--913}.
\newblock
\showISBNx{9781450383912}
\urldef\tempurl%
\url{https://doi.org/10.1145/3453483.3454084}
\showDOI{\tempurl}


\bibitem[Terei et~al\mbox{.}(2012)]%
        {safe-haskell}
\bibfield{author}{\bibinfo{person}{David Terei}, \bibinfo{person}{Simon
  Marlow}, \bibinfo{person}{Simon Peyton~Jones}, {and} \bibinfo{person}{David
  Mazi\`{e}res}.} \bibinfo{year}{2012}\natexlab{}.
\newblock \showarticletitle{Safe {H}askell}.
\newblock \bibinfo{journal}{\emph{SIGPLAN Notices (Proceedings of the 2012
  Haskell Symposium)}} \bibinfo{volume}{47}, \bibinfo{number}{12}
  (\bibinfo{date}{Sept.} \bibinfo{year}{2012}), \bibinfo{pages}{137–148}.
\newblock
\showISSN{0362-1340}
\urldef\tempurl%
\url{https://doi.org/10.1145/2430532.2364524}
\showDOI{\tempurl}


\bibitem[Tolnay(2025)]%
        {tolnay2024impression}
\bibfield{author}{\bibinfo{person}{David Tolnay}.}
  \bibinfo{year}{2025}\natexlab{}.
\newblock \bibinfo{title}{Tweet on January 27, 2025 (topic: {AI}-maintained
  {serde\_yml})}.
\newblock
  \bibinfo{howpublished}{\url{https://x.com/davidtolnay/status/1883906113428676938}}.
\newblock
\newblock
\shownote{Accessed: 2025-03-25}.


\bibitem[Toone et~al\mbox{.}(2003)]%
        {toone2003trust}
\bibfield{author}{\bibinfo{person}{Brian Toone}, \bibinfo{person}{Michael
  Gertz}, {and} \bibinfo{person}{Premkumar Devanbu}.}
  \bibinfo{year}{2003}\natexlab{}.
\newblock \showarticletitle{Trust Mediation for Distributed Information
  Systems}. In \bibinfo{booktitle}{\emph{Security and Privacy in the Age of
  Uncertainty}}, \bibfield{editor}{\bibinfo{person}{Dimitris Gritzalis},
  \bibinfo{person}{Sabrina De~Capitani~di Vimercati},
  \bibinfo{person}{Pierangela Samarati}, {and} \bibinfo{person}{Sokratis
  Katsikas}} (Eds.). \bibinfo{publisher}{Springer US},
  \bibinfo{address}{Boston, MA}, \bibinfo{pages}{1--12}.
\newblock
\showISBNx{978-0-387-35691-4}


\bibitem[Umans(2001)]%
        {umans2001implicants}
\bibfield{author}{\bibinfo{person}{Christopher Umans}.}
  \bibinfo{year}{2001}\natexlab{}.
\newblock \showarticletitle{The Minimum Equivalent {DNF} Problem and Shortest
  Implicants}.
\newblock \bibinfo{journal}{\emph{J. Comput. System Sci.}}
  \bibinfo{volume}{63}, \bibinfo{number}{4} (\bibinfo{year}{2001}),
  \bibinfo{pages}{597--611}.
\newblock
\showISSN{0022-0000}
\urldef\tempurl%
\url{https://doi.org/10.1006/jcss.2001.1775}
\showDOI{\tempurl}


\bibitem[user)(2024)]%
        {cve-rs}
\bibfield{author}{\bibinfo{person}{Speykious~({GitHub} user)}.}
  \bibinfo{year}{2024}\natexlab{}.
\newblock \bibinfo{title}{{cve-rs}: {B}lazingly fast memory vulnerabilities,
  written in 100\% safe {R}ust.}
\newblock \bibinfo{howpublished}{\url{https://github.com/Speykious/cve-rs}}.
\newblock
\newblock
\shownote{[Accessed 11-11-2024]}.


\bibitem[VanHattum et~al\mbox{.}(2022)]%
        {kani}
\bibfield{author}{\bibinfo{person}{Alexa VanHattum}, \bibinfo{person}{Daniel
  Schwartz-Narbonne}, \bibinfo{person}{Nathan Chong}, {and}
  \bibinfo{person}{Adrian Sampson}.} \bibinfo{year}{2022}\natexlab{}.
\newblock \showarticletitle{Verifying dynamic trait objects in {R}ust}. In
  \bibinfo{booktitle}{\emph{Proceedings of the 44th International Conference on
  Software Engineering: Software Engineering in Practice}} (Pittsburgh,
  Pennsylvania) \emph{(\bibinfo{series}{ICSE-SEIP '22})}.
  \bibinfo{publisher}{Association for Computing Machinery},
  \bibinfo{address}{New York, NY, USA}, \bibinfo{pages}{321--330}.
\newblock
\showISBNx{9781450392266}
\urldef\tempurl%
\url{https://doi.org/10.1145/3510457.3513031}
\showDOI{\tempurl}


\bibitem[Villani et~al\mbox{.}(2025)]%
        {tree-borrows}
\bibfield{author}{\bibinfo{person}{Neven Villani}, \bibinfo{person}{Johannes
  Hostert}, \bibinfo{person}{Derek Dreyer}, {and} \bibinfo{person}{Ralf Jung}.}
  \bibinfo{year}{2025}\natexlab{}.
\newblock \showarticletitle{Tree borrows}.
\newblock \bibinfo{journal}{\emph{Proceedings of the ACM on Programming
  Languages}} \bibinfo{volume}{9}, \bibinfo{number}{PLDI}
  (\bibinfo{year}{2025}), \bibinfo{pages}{1019--1042}.
\newblock


\bibitem[Wagner et~al\mbox{.}(2025)]%
        {wagner2025linearity}
\bibfield{author}{\bibinfo{person}{Andrew Wagner}, \bibinfo{person}{Olek
  Gierczak}, \bibinfo{person}{Brianna Marshall}, \bibinfo{person}{John~M Li},
  {and} \bibinfo{person}{Amal Ahmed}.} \bibinfo{year}{2025}\natexlab{}.
\newblock \showarticletitle{From Linearity to Borrowing}.
\newblock \bibinfo{journal}{\emph{Proceedings of the ACM on Programming
  Languages}} \bibinfo{volume}{9}, \bibinfo{number}{OOPSLA2}
  (\bibinfo{year}{2025}), \bibinfo{pages}{3981--4007}.
\newblock


\bibitem[Wang(2021)]%
        {wang2021feasibility}
\bibfield{author}{\bibinfo{person}{Xinyuan Wang}.}
  \bibinfo{year}{2021}\natexlab{}.
\newblock \showarticletitle{On the Feasibility of Detecting Software Supply
  Chain Attacks}. In \bibinfo{booktitle}{\emph{MILCOM 2021 - 2021 IEEE Military
  Communications Conference (MILCOM)}}. \bibinfo{publisher}{IEEE},
  \bibinfo{address}{Piscataway, NJ, USA}, \bibinfo{pages}{458--463}.
\newblock
\urldef\tempurl%
\url{https://doi.org/10.1109/MILCOM52596.2021.9652901}
\showDOI{\tempurl}


\bibitem[Weiss et~al\mbox{.}(2019)]%
        {oxide}
\bibfield{author}{\bibinfo{person}{Aaron Weiss}, \bibinfo{person}{Olek
  Gierczak}, \bibinfo{person}{Daniel Patterson}, {and} \bibinfo{person}{Amal
  Ahmed}.} \bibinfo{year}{2019}\natexlab{}.
\newblock \bibinfo{title}{{O}xide: The Essence of {R}ust}.
\newblock \bibinfo{howpublished}{arXiv}.
\newblock
\showeprint[arxiv]{1903.00982}


\bibitem[Xia et~al\mbox{.}(2023)]%
        {xia2023empirical}
\bibfield{author}{\bibinfo{person}{Boming Xia}, \bibinfo{person}{Tingting Bi},
  \bibinfo{person}{Zhenchang Xing}, \bibinfo{person}{Qinghua Lu}, {and}
  \bibinfo{person}{Liming Zhu}.} \bibinfo{year}{2023}\natexlab{}.
\newblock \showarticletitle{An Empirical Study on Software Bill of Materials:
  Where We Stand and the Road Ahead}. In \bibinfo{booktitle}{\emph{Proceedings
  of the 45th International Conference on Software Engineering}} (Melbourne,
  Victoria, Australia) \emph{(\bibinfo{series}{ICSE '23})}.
  \bibinfo{publisher}{IEEE Press}, \bibinfo{address}{Los Alamitos, CA, USA},
  \bibinfo{pages}{2630--2642}.
\newblock
\showISBNx{9781665457019}
\urldef\tempurl%
\url{https://doi.org/10.1109/ICSE48619.2023.00219}
\showDOI{\tempurl}


\bibitem[Zahan et~al\mbox{.}(2023)]%
        {zahan2023software}
\bibfield{author}{\bibinfo{person}{Nusrat Zahan}, \bibinfo{person}{Elizabeth
  Lin}, \bibinfo{person}{Mahzabin Tamanna}, \bibinfo{person}{William Enck},
  {and} \bibinfo{person}{Laurie Williams}.} \bibinfo{year}{2023}\natexlab{}.
\newblock \showarticletitle{Software bills of materials are required. Are we
  there yet?}
\newblock \bibinfo{journal}{\emph{IEEE Security \& Privacy}}
  \bibinfo{volume}{21}, \bibinfo{number}{2} (\bibinfo{year}{2023}),
  \bibinfo{pages}{82--88}.
\newblock


\bibitem[Zahan et~al\mbox{.}(2022)]%
        {zahan2022weak}
\bibfield{author}{\bibinfo{person}{Nusrat Zahan}, \bibinfo{person}{Thomas
  Zimmermann}, \bibinfo{person}{Patrice Godefroid}, \bibinfo{person}{Brendan
  Murphy}, \bibinfo{person}{Chandra Maddila}, {and} \bibinfo{person}{Laurie
  Williams}.} \bibinfo{year}{2022}\natexlab{}.
\newblock \showarticletitle{What are weak links in the {npm} supply chain?}. In
  \bibinfo{booktitle}{\emph{Proceedings of the 44th International Conference on
  Software Engineering: Software Engineering in Practice}} (Pittsburgh,
  Pennsylvania) \emph{(\bibinfo{series}{ICSE-SEIP '22})}.
  \bibinfo{publisher}{Association for Computing Machinery},
  \bibinfo{address}{New York, NY, USA}, \bibinfo{pages}{331--340}.
\newblock
\showISBNx{9781450392266}
\urldef\tempurl%
\url{https://doi.org/10.1145/3510457.3513044}
\showDOI{\tempurl}


\bibitem[Zeller and Snelting(1997)]%
        {zeller1997unified}
\bibfield{author}{\bibinfo{person}{Andreas Zeller} {and}
  \bibinfo{person}{Gregor Snelting}.} \bibinfo{year}{1997}\natexlab{}.
\newblock \showarticletitle{Unified versioning through feature logic}.
\newblock \bibinfo{journal}{\emph{ACM Transactions on Software Engineering and
  Methodology (TOSEM)}} \bibinfo{volume}{6}, \bibinfo{number}{4}
  (\bibinfo{year}{1997}), \bibinfo{pages}{398--441}.
\newblock


\bibitem[Zimmermann et~al\mbox{.}(2019)]%
        {zimmermann2019small}
\bibfield{author}{\bibinfo{person}{Markus Zimmermann},
  \bibinfo{person}{Cristian-Alexandru Staicu}, \bibinfo{person}{Cam Tenny},
  {and} \bibinfo{person}{Michael Pradel}.} \bibinfo{year}{2019}\natexlab{}.
\newblock \showarticletitle{Small World with high risks: a study of security
  threats in the {npm} ecosystem}. In \bibinfo{booktitle}{\emph{Proceedings of
  the 28th USENIX Conference on Security Symposium}} (Santa Clara, CA, USA)
  \emph{(\bibinfo{series}{SEC'19})}. \bibinfo{publisher}{USENIX Association},
  \bibinfo{address}{USA}, \bibinfo{pages}{995--1010}.
\newblock
\showISBNx{9781939133069}


\bibitem[Zoghbi et~al\mbox{.}(2026)]%
        {cargo-scan-paper}
\bibfield{author}{\bibinfo{person}{Lydia Zoghbi}, \bibinfo{person}{David
  Thien}, \bibinfo{person}{Ranjit Jhala}, \bibinfo{person}{Deian Stefan}, {and}
  \bibinfo{person}{Caleb Stanford}.} \bibinfo{year}{2026}\natexlab{}.
\newblock \bibinfo{title}{Auditing {R}ust Crates Effectively}.
\newblock \bibinfo{howpublished}{To appear, experience report in European
  Symposium on Programming (ESOP)}.
\newblock


\bibitem[Zorz(2024)]%
        {helpnetsecurityBewareBackdoor}
\bibfield{author}{\bibinfo{person}{Zeljka Zorz}.}
  \bibinfo{year}{2024}\natexlab{}.
\newblock \bibinfo{title}{{B}eware! {B}ackdoor found in {X}{Z} utilities used
  by many {L}inux distros ({C}{V}{E}-2024-3094) - {H}elp {N}et {S}ecurity}.
\newblock
  \bibinfo{howpublished}{\url{https://www.helpnetsecurity.com/2024/03/29/cve-2024-3094-linux-backdoor/}}.
\newblock
\newblock
\shownote{[Accessed 05-11-2025]}.


\end{thebibliography}
